\documentclass[11pt]{article}

\usepackage[english]{babel}
\usepackage[utf8]{inputenc}
\usepackage{amsmath, amssymb, geometry}
\usepackage{tikz}
\usepackage{bm}
\usepackage{quantikz}
\usepackage[backref=page, colorlinks=true, allcolors=blue] {hyperref}
\usepackage{cleveref}
\usepackage{float}
\usepackage{authblk}
\usepackage{amsthm}
\usepackage{palatino}
\usepackage{mathpazo}
\usepackage{booktabs}

\usepackage{subcaption}

\newcommand{\tp}{\otimes}
\newcommand{{\real}}{\mathbb{R}}
\newcommand{{\expected}}{\mathbb{E}}
\newcommand{\ip}[2]{\left\langle #1 , #2\right\rangle}
\newcommand{\tr}[1]{\mathrm{Tr} \left( #1 \right)}

\newtheorem{theorem}{Theorem}
\newtheorem{lemma}{Lemma}

\theoremstyle{remark}
\newtheorem*{remark}{Remark}

\hypersetup{pdfpagemode=UseNone}

\newenvironment{mylist}[1]{\begin{list}{}{
	\setlength{\leftmargin}{#1}
	\setlength{\rightmargin}{0mm}
	\setlength{\labelsep}{2mm}
	\setlength{\labelwidth}{8mm}
	\setlength{\itemsep}{0mm}}}
	{\end{list}}
	
\newcommand{\R}{\mathbb{R}}
\newcommand{\abs}[1]{\left|#1\right|}

\newcommand{\dd}{\mathrm{d}}

\newcommand{\lrb}[1]{\left ( #1 \right )}

\newcommand{\mysin}[1]{\operatorname{sin}\lrb{#1}}
\newcommand{\mycos}[1]{\operatorname{cos}\lrb{#1}}

\newcommand{\myexp}[1]{\operatorname{exp} \lrb{#1}}

\newcommand{\polylog}[1]{\operatorname{polylog} \lrb{#1}}

\newcommand{\mylog}[1]{\operatorname{log} \lrb{#1}}

\newcommand{\defeq}{\coloneqq}

\DeclareMathOperator{\erfc}{erfc}

\usepackage{algpseudocode}
\usepackage{mdframed}

\begin{document}

\title{Exponentially Fast Solution State Preparation for the Heat Equation and its use for Option Pricing}

\author[1]{Gumaro Rend\'on\thanks{grendon@fujitsu.com}} 
\affil[1]{Fujitsu Research of America, Inc, Santa Clara, CA 95054, USA}

\author[2,3]{\v{S}t\v{e}p\'an {\v{S}}m{\'{i}}d}
\affil[2]{Department of Computing, Imperial College London, United Kingdom}
\affil[3]{Fujitsu Research of Europe Ltd., Slough SL1 2BE, United Kingdom}

\author[1]{Sarvagya Upadhyay}

\maketitle

\begin{abstract}
    In this work, we present the methods necessary to price an important set of derivatives on a quantum device while offering an advantage over existing classical methods. The methods developed here, in conjunction with ~\cite{GumaroS2026}, also provide an exponential advantage in requirement of qubits when pricing some option contracts with path-dependent payoff compared to state-of-the-art quantum Monte Carlo methods.
\end{abstract}

\pagebreak

\tableofcontents

\pagebreak

\section{Introduction}

Quantum computing traces its inception to early 1980s as an inspiration to simulate quantum systems on a physical device that follows the postulates of quantum mechanics. This insight led to significant development in this field from computation viewpoint, most notably, formalization of a quantum computational models, Shor's algorithm for prime factorization and discrete-log problem, Grover's search algorithm of unstructured databases, and quantum error-correction. Simultaneously, there have been significant development in quantum cryptography and communication via quantum information.

Over the past decade, the theoretical developments of preceding years have been complemented by substantial progress toward making quantum computing practical. Technological breakthroughs in control, fabrication, and error mitigation have led to various hardware proposals such as superconducting, trapped ions, neutral-atoms, diamond-spin and photonic platforms. Each offer distinct tradeoffs in connectivity, coherence time, speed of gate execution, and scalability. Alongside progress across various hardware modalities. there have been rapid progress in developing quantum algorithmic applications across disciplines such as quantum simulations in chemistry and materials, optimization, machine learning, and finance. Given the current hardware constraints, the application developments have forayed beyond foundational algorithms that provide rigorous guarantees to include variational quantum algorithms (VQAs) and QAOA. 

While there are merits to investigate quantum algorithms tailored to so-called NISQ devices, it is universally believed that true quantum advantage can be realized only with a scalable fault-tolerant quantum computing platform. Quantum algorithms that would run on such devices are typically constructed from a small collection of fundamental primitives that recur across a wide range of applications. Among these, quantum phase estimation, quantum amplitude estimation, polynomial transformations of operators, and Hamiltonian simulation play a central role. Together they have been utilized as subroutines into general frameworks that support a broad class of algorithms for numerical computation.

Quantum phase estimation (QPE) enables the extraction of spectral information from a unitary operator and serves as a foundational component in algorithms for eigenvalue problems, order finding, and spectral analysis. Variants of QPE appear in quantum algorithms for linear systems, differential equations, and optimization, where access to eigenvalues is essential \cite{AbramsL1997,HarrowHL2009}. Closely related ideas also underlie quantum signal processing techniques, in which controlled evolutions encode polynomial functions of an operator’s spectrum. On the other hand, quantum amplitude estimation (QAE) generalizes Grover’s search procedure and provides a quadratic improvement over classical sampling for estimating expectations of bounded random variables \cite{BrassardHMT2002}. QAE is widely used as a subroutine when one wishes to extract a scalar quantity from a quantum state. For instance, QAE is often appears after a quantum state encoding a solution or distribution has been prepared \cite{Montanaro2015}. A unifying abstraction for many modern quantum algorithms is quantum singular value transformation (QSVT). It provides a systematic method for applying polynomial transformations to the singular values of matrices embedded into unitary operators via block-encodings \cite{GilyenSLW2019}. This framework subsumes earlier constructions such as quantum signal processing and qubitization \cite{LowC2017,LowC2019} and has become a standard tool for expressing quantum algorithms for linear algebra, Hamiltonian simulation, and differential equations.

\paragraph{Quantum numerical algorithms} 
These fundamental quantum algorithms can be repurposed for various numerical algorithms. If the numerical method can be reduced to linear algebra such as solving linear systems, applying matrix functions, simulating linear dynamics, or estimating expectations, then above framework can be utilized to work with complexity that scales only polylogarithmically with the dimension. The caveat is that the output is usually encoded in a quantum state. This feature is both the source of potential quantum advantage and challenge if we wish to obtain the correct solution classically.

For various applications involving numerical algorithms and simulations, there are three main sources of potential quantum advantages.

\begin{mylist}{\parindent}
\item [1.] {\bf Quantum Hamiltonian simulation.}
Hamiltonian simulation asks to implement $\exp(-iHt)$ for a Hamiltonian $H$ and time $t$. Over the years, a sequence of increasingly optimal frameworks emerged: sparse-Hamiltonian simulation and product formulas, quantum walks, truncated Taylor series methods, and finally the near-optimal modern viewpoint built around \emph{quantum signal processing} (QSP) and \emph{qubitization} \cite{ChildsK2011, BerryCCKS2015, LowC2017, LowC2019}. These ideas were unified and generalized through block-encoding and QSVT, which allow one to apply polynomial transformations to the singular values of matrices embedded in unitaries \cite{GilyenSLW2019}. From a numerical perspective, Hamiltonian simulation can be interpreted as a method for applying matrix exponentials or related functions to quantum states. This interpretation is particularly relevant for differential equations, where formal solutions often involve exponential of linear operators.

\item [2.]{\bf Quantum linear algebra and linear systems.}
A watershed result was the HHL algorithm for linear system solver. For a matrix equation $Ax=b$, it outputs a quantum state proportional to the solution in time polylogarithmic in dimension and dependence on sparsity and condition number of $A$ \cite{HarrowHL2009}. Subsequent work clarified how preconditioning, conditioning, and precision interact, and improved the dependence on the error tolerance from polynomial in $1/\epsilon$ to polylogarithmic in $1/\epsilon$ for broad families of problems \cite{CladerJS2013, ChildsKS2017}. These developments are important for  ordinary and partial differential equations because most deterministic discretizations (finite differences, finite elements, spectral methods) reduce solving ODEs or PDEs to repeatedly applying sparse structured matrices and solving large sparse linear systems.

\item [3.]{\bf Quantum estimation and Monte Carlo speedups.}
In parallel, there are quantum algorithmic techniques for estimating expectations and probabilities. Amplitude amplification and amplitude estimation yield a near-quadratic speedup for many Monte Carlo-type tasks \cite{BrassardHMT2002}. This perspective was formalized by \cite{Montanaro2015}, who showed broad settings in which quantum algorithms can accelerate Monte Carlo estimation, and it has become central to quantum algorithms in financial sector, where derivative pricing mostly boils down to estimating expectation.
\end{mylist}

These three broad techniques have been widely used in broad range of applications across disciplines. One such avenue that is highly important from a quantum impact viewpoint and technically rich to warrant our attention is the theory of derivative pricing.

\paragraph{Derivative pricing theory} 
The main goal that is addressed in this work is to develop more efficient quantum numerical algorithms for pricing financial derivatives. Derivatives are financial products that derive their value from an underlying asset such as stock, foreign-exchange (FX) rates, indices, interest-rates (IR), etc. Correctly evaluating their value heavily depends on the modeling choices of the underlying asset and how the parameters of the model behave. Once the model is established and the parameters behavior is determined, practitioners approach the valuation problem from two different computational methods: using Monte Carlo method and its variants or using partial differential equations (PDEs). In this paper, we will focus on three different categories from a PDE perspective: 

\begin{itemize}
    \item {\bf Equity derivatives.} In equity markets, the Black--Scholes model provides a canonical example in which European option prices satisfy a parabolic PDE that can be transformed into a heat equation \cite{BlackS1973,Merton1973}. Extensions to multi-asset options lead to higher-dimensional diffusion equations with correlated drivers. Bermudan options introduce boundary and interface conditions that further complicate numerical treatment \cite{Wilmott2006,Haug2007}. While Black--Scholes assumes constant or time-dependent volatility, more sophisticated models such as local volatility \cite{Dupire1994,DermanK1994} and stochastic volatility models including Heston and SABR are widely used in practice \cite{Heston1993, HaganKLW2002}. These models generally lead to PDEs with state-dependent coefficients or additional dimensions.

    \item {\bf Interest-rate derivatives.} Interest-rate derivatives, such as caps, floors, and swaptions, are commonly priced using either short-rate models or market models. Short-rate models such as Vasicek and Hull--White describe the evolution of the instantaneous short rate and lead to low-dimensional PDEs \cite{Vasicek1977, HullW1990, BrigoM2006}. Market models, including the LIBOR Market Model, describe the dynamics of forward rates directly and are consistent with observed caplet and swaption volatilities \cite{BraceGM1997, GlassermanZ2000}. Under appropriate measure choices and standard approximations, swaption pricing problems in these models can be reduced to parabolic PDEs with time-dependent coefficients.

    \item {\bf Cross-asset derivatives.} Cross-asset products combine risk factors across two or more asset classes, such as equity--interest-rate or FX--interest-rate derivatives. Pricing them typically involves multi-dimensional diffusion processes with correlated drivers \cite{Bjork2009,BrigoM2006}. Such models often combine features of equity, FX, and interest-rate dynamics, leading to PDEs whose dimensionality and structure depend on the modeling choices. While these models are more complex, they still admit parabolic PDE formulations under standard assumptions for certain financial products.
\end{itemize}

It is worth noting that Monte Carlo methods are widely used in practice due to their flexibility and robustness in high-dimensional settings, particularly for path-dependent payoffs. PDE solvers, by contrast, propagate value functions directly and are often preferred in low to moderate dimensions. Classical numerical practice relies heavily on both approaches, with the choice dictated by dimensionality, payoff structure, and accuracy requirements \cite{Wilmott2006, Bjork2009}.

Classically, PDE-based methods remain an important benchmark for pricing and risk analysis, particularly when they are tractable and full solution surfaces and sensitivities are required. The applications discussed in this paper focus on state-of-the-art models where the resulting PDEs possess sufficient structure to be amenable to faster quantum algorithms even when the dimension of the PDEs are large enough to make them classically intractable (from a practical viewpoint).

Coming back to quantum algorithms in derivative pricing theory, several  algorithms have been developed that are based on quantum versions of Monte Carlo simulation. These include direct applications of quantum amplitude estimation to payoff sampling, as well as more sophisticated approaches incorporating multilevel Monte Carlo and regression-based techniques for early exercise features \cite{Montanaro2015, StamatopoulosESZISW2020, WoernerE2019, Heinrich2002, Giles2008}. Such methods can offer quadratic improvements in sampling complexity, but they inherit the structural characteristics of Monte Carlo approaches, including the need to simulate entire paths and, for early-exercise products, to approximate continuation values as will be discussed later on.

By contrast, relatively little work has focused on quantum algorithms for PDE-based derivative pricing. As mentioned earlier, PDE solvers operate on value functions rather than individual paths and naturally produce entire solution surfaces classically, from which prices and sensitivities can be obtained simultaneously. When the underlying models lead to parabolic PDEs with structured generators, PDE-based quantum algorithms offer a qualitatively different approach from quantum Monte Carlo. The present work focuses on this latter direction from a Hamiltonian simulation perspective.

\noindent \paragraph{Differential equations and Hamiltonian simulation}
A common differential equation one works with is 
\[
\frac{d}{dt} \ket{u(t)} = -A \ket{u(t)} \qquad \implies \qquad \ket{u(t)} = e^{-tA}\ket{u(0)},
\]
where $A$ is a linear operator arising from discretization. Implementing the evolution operator $\exp(-tA)$, or a suitable approximation, is therefore a central task. When $A$ admits a block-encoding and possesses additional structure, Hamiltonian simulation and QSVT provide mechanisms for approximating this evolution efficiently.

Modern Hamiltonian simulation algorithms achieve near-optimal asymptotic scaling in time, precision, and operator norm for broad classes of structured Hamiltonians. A central complexity parameter in Hamiltonian simulation is the product of the evolution time $t$ and an appropriate norm of the operator $H$. For generic Hamiltonians, simulation cost scales at least linearly in this parameter. The possibility of simulating evolution for time $t$ in
sublinear or even polylogarithmic in time is referred to as \emph{fast-forwardability}. Given these near-optimality results, a fundamental question one can ask is how to circumvent them. Clearly, for general Hamiltonians, the simulation cost scales at least linearly with the evolution time and an appropriate norm of the operator. Results ruling out generic fast-forwarding demonstrate that sublinear time simulation cannot be achieved without exploiting special structure \cite{ChildsK2011, ChiaCHLLS2023} under plausible complexity-theoretic assumptions. Consequently, algorithms that obtain improved scaling rely on properties such as sparsity, commutativity, or diagonalization in a known basis.

In many ODEs or PDEs, operators derived from diffusion or parabolic equations are diagonalizable in the Fourier basis or can be decomposed into commuting components. Such structure can be exploited to reduce simulation costs relative to generic sparse operators. However, the output of these procedures is typically a quantum state, and extracting useful classical information requires additional processing. Any assessment of algorithmic performance must therefore include both state preparation and extraction costs.

In fact, one can utilize the discretization approach as well. Discretizations of ordinary and partial differential equations typically yield large structured linear operators whose action must be applied repeatedly. Early quantum algorithms addressed boundary-value problems such as the Poisson equation by mapping them to sparse linear systems \cite{CaoPPTK2013}. Subsequent work developed high-precision algorithms for linear ordinary differential equations with improved dependence on the target accuracy \cite{BerryCOW2017, ChildsL2020}. For partial differential equations, more recent results analyze discretization error, conditioning, and precision in a unified framework \cite{ChildsLO2020}. Related studies of quantum algorithms for the finite element method emphasize that discretization complexity and smoothness assumptions play a critical role in determining whether quantum algorithms can provide asymptotic improvements over classical solvers \cite{MontanaroP2016}.

\noindent \paragraph{Our Contributions}
This paper develops a quantum algorithmic framework for solving heat-equation-type PDEs and applies it formulations arising in derivative pricing. The contributions of this work can be summarized as follows:

\begin{mylist}{\parindent}
\item [1.] We construct a high-precision quantum procedure for implementing diffusion-type evolution operators using Fourier-based LCU techniques, exploiting structural properties that enable improved time scaling relative to generic Hamiltonian simulation. 

\item [2.] The method we propose for implementing the diffusion evolution operator allows for a wider range of initial conditions than previously implemented (refer to \cite{LindenMS2022,LubaschKWM2025}).

\item [3.] We explicitly leverage fast-forwardable components of discretized diffusion operators, clarifying how Fourier diagonalization and commuting structures allow favorable dependence on the evolution time under realistic assumptions.

\item[4.] We also obtain an implementation of the advection/drift term whose error is exponentially vanishing with respect to resources even in the presence of kinks in the initial conditions.

\item [5.] We formulate and analyze applications to derivative pricing across equity, interest-rate, and hybrid products, focusing on models whose PDE representations admit structured generators amenable to quantum numerical methods.

\item [6.] We compare the proposed approach  to existing quantum Monte Carlo methods for financial derivatives, highlighting differences in modeling assumptions, output structure, and resource requirements.
\end{mylist}

Together, these results contribute to the growing body of work on quantum algorithms for numerical computation by identifying a structured and practically motivated class of PDEs for which Hamiltonian-simulation-based methods provide a viable alternative to Monte Carlo-based approaches.

\vspace{10mm}

\medskip
\noindent\paragraph{From quantum speedups to end-to-end pricing.}
The discussion above highlights a recurring theme in quantum numerical algorithms: one often obtains strong asymptotic improvements for \emph{preparing} a quantum state that encodes the solution of a computational problem (e.g., a value function for a PDE), while the quantity of practical interest---such as an option price or an entire solution surface---is ultimately a \emph{classical} object. Consequently, any credible claim of advantage for derivative pricing must be assessed end-to-end, accounting not only for the cost of implementing the relevant evolution (or solver), but also for the cost of \emph{extracting} the desired information from the prepared quantum state.

In this paper we focus on PDE-based pricing routes that reduce to heat-equation-type dynamics, where the generator has additional structure that can be exploited algorithmically. This motivates two coupled technical questions that guide the remainder of the introduction:
(i) can one implement the diffusion (and related) evolution operators with better-than-linear scaling in the natural simulation parameter $\tau = t\|A\|$ by exploiting fast-forwardable structure; and
(ii) once the solution is encoded as amplitudes on a grid, can one recover the relevant classical information without destroying the speedup by na\"{i}ve pointwise amplitude estimation?
The next paragraphs position our work relative to existing quantum Monte Carlo approaches and prior PDE-based quantum algorithms, and then summarize the extraction guarantee we rely on from~\cite{GumaroS2026}.

There has been substantial attention on the use of quantum computers to obtain a speed-up in pricing financial instruments. Closed form expressions to solve these problems exist to the simplest cases (see for example, European option contracts). The stochastic nature of these problems, that is, stochastic differential equations (SDE), allow for the use of Monte Carlo techniques for an approximate solution of these problems. In this specific approach, there is the possibility of speeding up this type of calculations using Quantum amplitude estimation techniques. A quadratic speed-up is expected by direct application of this quantum amplitude estimation methods to Monte Carlo techniques for pricing. Although a more precise assessment of quantum speed-up over classical on a particular problem must be made by comparing with deterministic pricing methods as well like partial differential equation (PDE) solving.

This brings us to quantum methods for PDE solving. The seminal work of authors in \cite{HarrowHL2009} provide an exponential speed-up with respect to the problem size when preparing the solution to a linear system. The polynomial dependence of the cost with respect to the condition number $\kappa$ persists. This condition number often hides the polynomial dependence with respect to the system size. Other methods exist to solve PDEs like through quantum signal processing~\cite{GilyenSLW2019} or Schr\"odingerization~\cite{JinLY2024}, but depend linearly with respect to $\tau = \| A \| t$. In the case of the heat equation or other derivative operators, the norm $\|A\|$ grows at least linearly with the system size (number of grid points). There might be a way to improve this if the operator $A$ is fast-forwardable, for example, \cite{AnOY2024} achieves $\sqrt{t \| A \|}$ (ignoring $\mathcal{O}\left({\rm polylog}(1/\epsilon)\right)$ factors, where $\epsilon$ is the error on the state-vector) scaling. In this work, we improve this by achieving ${\rm polylog}(t\| A\|,1/\epsilon)$. For path dependant option contracts, Quantum Monte Carlo methods require the storage of amplitudes for all path combinations which means one needs ${\rm poly}(1/\epsilon)$ number of qubits \cite{PrakashSCCDHKSWKo2024}. In contrast, quantum PDE solving methods offer a more efficient use of quantum memory improving this requirement to ${\rm polylog}(1/\epsilon)$.

In each of these quantum algorithms for solving a PDE, the target error is always calculated with respect to the state-vector prepared in quantum memory. However, we are interested in obtaining the full-form of the solution and by using a naive implementation of quantum amplitude estimation (QAE) for extracting the solution at each point of the grid negates the exponential quantum speed-ups as pointed out by authors in \cite{HarrowHL2009}.

First, the work in \cite{Rendon2025} showcases a way to avoid this problem under mild assumptions (positive definiteness and smoothness) on the solution function. This was for a 1D problem and the extraction method scaled with the solution condition number: $\mathcal{O}\left(\frac{\max_x\psi(x)}{\min_x\psi(x)}\right)$.  In the subsequent work ~\cite{GumaroS2026} this is extended to several variables and remove the condition number scaling with through interference with a constant distribution. The results can be summarized in the following lemma:

\begin{lemma}[Informal: Approximated quantum approximation of smooth functions]
\label{lem:SE_multidim}
Let $\psi : [-1,1]^D \to \mathbb{C}$ be a smooth function satisfying
\[
\int_{[-1,1]^D} |\psi(y)|^2 \, \dd y = 1,
\]
and suppose all its derivatives are uniformly bounded by a constant
$\Lambda$, i.e.,
\[
\|\partial^\alpha \psi\|_\infty \le \Lambda^{|\alpha|+1}
\quad \text{for all multi-indices } \alpha.
\]

Assume we are given quantum access to discretized samples of $\psi$
via a state-preparation procedure that succeeds with amplitude
$a_\psi$, using $n$ qubits per dimension and a uniform grid of size
$2^n$ in each coordinate.

Then, for any target accuracy $\epsilon_{\rm total}$, the corresponding
quantum state encoding $\psi$ can be prepared or estimated using a
number of quantum gates scaling as
\[
\tilde O\!\left(
\frac{1}{a_\psi}
\frac{2^{D/2} L^{4D}}{\epsilon_{\rm total}}
\right),
\qquad
L = O\!\left(
\Lambda + \frac{1}{D}\log\frac{1}{\epsilon_{\rm total}}
\right).
\]

In particular, smoothness controls the polynomial overhead, while the
cost grows exponentially with the dimension and inversely with the
desired accuracy.
\end{lemma}

An important subset of these pricing problems can be formulated as solving a heat equation (or a closely related parabolic PDE). We focus on this class because it is analytically clean and, crucially, its differential operators admit discretizations that are especially well suited to quantum implementation. With an appropriate discretization, we can prove an exponential improvement in the number of qubits required to encode the solution state in quantum memory. Combined with the solution-extraction procedure of~\cite{GumaroS2026}, our end-to-end scaling matches that of~\cite{LindenMS2022} when the goal is to recover the full solution surface. At the same time, our framework is more flexible: it accommodates drift terms and time-dependent coefficients, and it remains accurate even when the initial condition has kinks (i.e., discontinuous first derivatives) and even discontinuities when only diffusive.

\section{Derivative pricing theory}
In this section, we will give a brief overview of the mathematics behind derivative pricing theory that starts with stochastic processes and stochastic differential equations (SDEs). This will be followed by the introduction to asset classes relevant to this work and the models that we employ for such asset classes. We will also state the modeling assumptions that is either rooted in practice in financial institutions or have been usually made when developing quantum algorithms for pricing financial derivatives.

\subsection{SDEs, PDEs, and martingale measures}
This section gives a concise overview of stochastic dynamics, and the machinery used in quantitative finance for pricing derivatives. Our objective is to give a basic understanding of the principles used that is essential for the derivatives we consider in the subsequent section without getting in to too much detail. We will first state our stochastic differential equation (SDE) model  and then explain how smooth functionals of that model evolve via It\^o's Lemma. It\^o's Lemma is crucial in understanging how the SDEs can be converted in to partial differential equations (PDEs) via Kolmogorov equations and the Feynman-Kac formula. This leads us to a discussion on the martingale property and how it's essential to the uniqueness of the solution to the PDE. Finally, we introduce the Girsanov's theorem and the change of num\`eraire principle, showing how they change the drift terms without touching the diffusion terms. We end this section with concrete numeraires in equity, foreign-exchange, and interest-rate asset classes.

\medskip
\noindent\textbf{State dynamics.} Let $(\Omega,\mathcal F,(\mathcal F_t)_{t\ge0},\mathbb P)$ be a filtered probability space and $W=(W^1,\dots,W^m)^\top$ be an $m$-dimensional Wiener process adapted to this filtration. In other words, the vector-values stochastic process $W$ satisfies the following conditions
\[
W_0 = 0 \quad \text{ almost surely} 
\qquad \text{and} \qquad
W_t-W_s \sim \mathcal N(0,(t-s)\Sigma) \quad \text{for all} \quad 0 \leq s \leq t.
\]
Here, $\Sigma \in \mathbb R^{m \times m}$ is a symmetric and positive semidefinite matrix succinctly represented as $\Sigma \succeq 0$ (known as the covariance matrix). Moreover, the paths are continuous and the increments over disjoint time intervals are independent. Having armed with $W$, we can define a $n$-dimensional state process $X=(X^1,\dots,X^n)^\top$ that evolves as 
\begin{equation}\label{eq:sde}
 dX_t=b(t,X_t)\,dt+\sigma(t,X_t)\,dW_t, \qquad X_0=x\in\mathbb R^n,
\end{equation}
where 
\[
b:[0,T]\times\mathbb R^n\to\mathbb R^n 
\qquad \text{and} \qquad
\sigma:[0,T]\times\mathbb R^n\to\mathbb R^{n\times m}
\]
are the drift vector and and volatility matrix. The diffusion matrix is $a(t,x):=\sigma(t,x)\sigma(t,x)^\top\in\mathbb S_+^n$. Under the usual Lipschitz/growth hypotheses, \eqref{eq:sde} has a unique strong solution \cite[Thm. 5.2.1]{Oksendal2003}.

\medskip
\noindent\textbf{Smoothness and boundedness.} 
For a function on time and spatial variables, $f:[0,T] \times \mathbb R^n \rightarrow \mathbb R$ denoted $f(t,x)$, we write $f\in C^{1,2}$ if the function is differentiable with respect to $t$ and twice-differentiable with respect to $x$ with continuous entries. In other words, $\partial_t f$, 
\[
\partial_t f \in \mathbb R \qquad \text{and} \qquad
\nabla_x f \in \mathbb R^n \qquad \text{and} \qquad 
\nabla_x^2 f \in \mathbb R^{n\times n}
\]
exist with continuous entries. We say that $f$ is bounded, if \[\sup_{(t,x)}|f(t,x)|<\infty.\]
If $f$ satisfies both the smoothness and the boundedness property as just described, we say that $f \in C^{1,2}_b$. 

Before we state the fundamental result in stochastic calculus, the It\^o's Lemma, let us recall the Frobenius inner product defined as
\[
\ip{A}{B} = \tr{A^\top B}
\qquad \text{for} \qquad
A, B \in \real^{n \times m}.
\]

\begin{theorem}[It\^o's Lemma; \cite{Oksendal2003}, Theorem  4.2.1]\label{thm:ito}
If $X$ solves \eqref{eq:sde} and $f\in C^{1,2}$, then
\begin{equation}\label{eq:ito}
 df = \left(\partial_t f + 
 \ip{\nabla_x f}{b}  + 
 \frac{1}{2}\ip{A}{\nabla_x^2 f}\right) dt + 
 \ip{\sigma^\top\nabla_x f}{dW_t}
\end{equation}
\end{theorem}

One of the immediate observation from equation \eqref{eq:ito} is that it isolates the deterministic part (the $dt$ terms) and the stochastic part (the $dW_t$ terms). Under the assumption that $X$ is a martingale, this decomposition forms the crux of derivative pricing theory in quantitative finance. It enables us to turn SDEs into deterministic PDEs as we describe below: apply It\^o to a \emph{candidate value function}, take conditional expectations, and use the martingale property to  have zero expectation.

\medskip
\noindent\textbf{Backward Kolmogorov equation.} Fix a function $g:\mathbb R^n\to\mathbb R$ that satisfies the smoothness and boundedness condition, i.e., $g\in C^{1,2}_b$. Now, consider the conditional expectation
\begin{equation}\label{eq:u}
 u(t,x):=\expected\left[g(X_T)\mid X_t=x\right].
\end{equation}
Let us define the process $Y_t := u(t, X_t)$. By definition, both $Y_t$ and $X_t$ are adapted to the same filtration, and hence by the towering property of conditional expectation,
\[
\expected\left[Y_t | \mathcal F_s\right] = \expected [u(t,X_t) | \mathcal F_s] = \expected [g(X_T) | \mathcal F_s] = u(s,X_s) = Y_s.
\]
This implies that the process $Y_t$ is a martingale, and hence, the drift term of the process $Y_t$ should vanish. Plugging $f=u$ into \eqref{eq:ito} and taking conditional expectations, we observe a drift term. Setting it to $0$ yields the \emph{backward Kolmogorov equation}.

\begin{theorem}[Backward Kolmogorov equation; \cite{Shreve2004}, Exercise 6.8]\label{thm:bk}
Under standard smoothness assumption, $u$ in \eqref{eq:u} solves
\begin{equation}\label{eq:bk}
\partial_t u + 
 \ip{\nabla_x u}{b}  + 
 \frac{1}{2}\ip{A}{\nabla_x^2 u} = 0, \qquad u(T,x)=g(x).
\end{equation}
\end{theorem}

\medskip
\noindent\textbf{Discounting and the pricing PDE (Feynman--Kac formula).} As we will see shortly, tradable assets require discounting at a (possibly state/time dependent) short rate $r(t,x)$. Define the discounted conditional expectation
\begin{equation}\label{eq:V}
 V(t,x):=\mathbb E\Big[\exp\left(-\int_t^T r(s,X_s)\,ds\right)\,g(X_T)\,\Big|\,X_t=x\Big].
\end{equation}
The next result tells us that $V$ is the unique classical solution to a linear parabolic PDE and conversely that such PDE solutions admit the probabilistic representation \eqref{eq:V}.

\begin{theorem}[Feynman--Kac formula; \cite{Shreve2004}, Theorem 6.4.1 and Section 6.6]\label{thm:fkac}
Under smoothness and boundedness assumptions on $b,A,r$ and $g$, the function $V$ in \eqref{eq:V} solves
\begin{equation}\label{eq:fkacPDE}
\partial_t V + 
 \ip{\nabla_x V}{b}  + 
 \frac{1}{2}\ip{A}{\nabla_x^2 V} - rV = 0, \qquad V(T,x)=g(x),
\end{equation}
with existence and uniqueness in the class of bounded $C_b^{1,2}$ functions.
\end{theorem}
Roughly speaking, the equation holds assuming that the following discounted process satisfied the martingale property
\begin{equation}\label{eq:M}
 M_t:=\exp\!\Big(-\!\int_0^t r(s,X_s)\,ds\Big)\,V(t,X_t).
\end{equation}
Applying It\^o's lemma to $V(t,X_t)$, we will see that the drift of $M_t$ is exactly the left-hand side of \eqref{eq:fkacPDE}. 

\medskip
\noindent\textbf{Girsanov theorem.} We will now turn to change of measure technique. Let $(\Omega,\mathcal F,(\mathcal F_t)_{t\ge0},\mathbb P)$ and $(\Omega,\mathcal F,(\mathcal F_t)_{t\ge0},\mathbb Q)$ are two filtered probability spaces such that there exists a predictable process $\theta_t$ such that
\[
dW^{\mathbb Q}_t = dW^{\mathbb P}_t + \theta_t dt.
\]
If the state process $X$ follows equation~\ref{eq:sde} under measure $\mathbb P$, then under the measure $\mathbb Q$, it will be described by the SDE
\[
dX_t= \left(b^{\mathbb P}(t,X_t) - \sigma(t,X_t) \theta_t\right)\,dt+\sigma(t,X_t)\,dW^{\mathbb Q}_t.
\]
It is not evidently clear if such a measure $\mathbb Q$ exists. If there exists a probability measure $\mathbb Q$ under which $W_t^{\mathbb Q}$ is again a Wiener process,
then the diffusion has a modified drift
\begin{equation}
b^{\mathbb Q}(t, X_t) = b^{\mathbb P}(t, X_t) - \sigma(t, X_t)\,\theta_t,
\end{equation}
while its volatility structure $\sigma$ remains unchanged. This is exactly what Girsanov's theorem states.

\begin{theorem}[Girsanov Theorem]
Let $W_t^{\mathbb P}$ be an $m$-dimensional Wiener process under $\mathbb P$, and let
$\theta_t$ be an adapted process satisfying the integrability condition
\[
\mathbb E\!\left[\exp\!\left(\tfrac{1}{2}\int_0^T \|\theta_s\|^2\,ds\right)\right] < \infty
\]
for every finite $T > 0$. 
Then there exists another probability measure $\mathbb Q$, equivalent to $\mathbb P$, 
under which the process
\[
W_t^{\mathbb Q} := W_t^{\mathbb P} + \int_0^t \theta_s\,ds
\]
is an $m$-dimensional Wiener process.
\end{theorem}
Girsanov's theorem asserts that a deterministic shift in the drift term of a Wiener process can be absorbed by switching to an equivalent probability measure under which the shifted process remains a Wiener process. It provides the mathematical foundation for changing the drift while preserving the volatility structure. In particular, a suitable choice of $\theta_t$ can change the process to be without drift, and consequently, a martingale. The way to do goes via the following steps:

\begin{mylist}{\parindent}
  \item [1.] Under the original measure $\mathbb P$, a stochastic process follows
  \[
  dX_t = b^{\mathbb P}(t,X_t)\,dt + \sigma(t,X_t)\,dW_t^{\mathbb P}.
  \]
  \item [2.] Choose an adapted process $\theta_t$ (the desired drift adjustment) and define a new stochastic process by
  \[
  W_t^{\mathbb Q} = W_t^{\mathbb P} + \int_0^t \theta_s\,ds.
  \]
  \item [3.] Postulate the existence of a measure $\mathbb Q$ under which $W_t^{\mathbb Q}$ is a Wiener process. This is where we make use of Girsanov's theorem.
  \item [4.] Substitute $dW_t^{\mathbb P} = dW_t^{\mathbb Q} - \theta_t\,dt$ into the SDE to obtain the drift-shifted dynamics:
  \[
  dX_t = (b^{\mathbb P} - \sigma(t,X_t)\theta_1)\,dt + \sigma(t,X_t)\,dW_t^{\mathbb Q}.
  \]
  \item [5.] Choose $\theta_t$ so that the stochastic process of interest has zero drift under $\mathbb Q$ making it a martingale.
\end{mylist}

We note that the martingale property simplifies the valuation and states that the current value of the derivative is such that no arbitrage opportunities arise.

\subsection{Asset classes and their modeling choices}
In this section we describe the stochastic models and specific modeling assumptions for each asset classes considered throughout the paper. We give a brief overview of change of numeraire framework introduced in the above section. Across all asset classes, the guiding principle is that asset prices, when expressed in an appropriate numeraire, are modeled as martingales. This guarantees absence of arbitrage and leads, via the Feynman--Kac representation, to linear parabolic pricing PDEs as described above. Our emphasis is on identifying modeling regimes in which these PDEs admit a reduction to heat-type equations with time-dependent but state-independent coefficients, which is essential for the quantum algorithms developed later in this work.

\subsubsection{Equity derivatives: Black--Scholes framework}

For equity derivatives we work within the Black--Scholes framework and its multi-asset extensions. Let $S_t$ denote the equity spot price. 
\[
B_t = \exp\!\left(\int_0^t r(u)\,du\right)
\]
be the domestic money-market account, which serves as the numeraire for the domestic risk-neutral measure $\mathbb{Q}^d$. Under this measure, the discounted equity price without dividends is assumed to be a martingale. The spot dynamics therefore take the form
\begin{equation}
 dS_t = S_t\left(r(t)\,dt + \sigma(t)\,dW_t^{\mathbb{Q}^d}\right),
\end{equation}
where $r(t)$ is the yield and $\sigma(t)$ is the volatility.

The volatility assumptions are tailored to the equity products considered in this paper. For multi-asset equity options, and Bermudan options with a fixed set of exercise dates, we assume constant volatilities. In the multi-asset setting this corresponds to a constant covariance matrix with fixed correlations across assets. This assumption is standard in both analytical and numerical treatments of early-exercise equity derivatives and leads to diffusion operators with constant coefficients, which are particularly amenable to orthogonal diagonalization and heat-equation transformations~\cite{BlackS1973, Bjork2009, Wilmott2006}.

For Bermudan equity options with longer maturities, we additionally allow for time-dependent volatility functions $\sigma(t)$ while excluding state dependence. This reflects common market practice, where volatility term structures are calibrated to vanilla options while maintaining tractability for exotic products. Importantly, time-dependent volatility preserves the parabolic structure of the pricing PDE and still allows a reduction to a heat-type equation after suitable changes of variables, as discussed in Appendix~C.

Under a change of numeraire to any strictly positive tradable process $N_t$, Girsanov's theorem implies that the diffusion term remains unchanged while the drift is modified by the covariance between the asset and the numeraire. This mechanism is repeatedly exploited in later sections to simplify or eliminate drift terms prior to transforming the PDE into heat-equation form.

\subsubsection{Interest-rate derivatives: LIBOR Market Model}

Interest-rate derivatives are modeled using the LIBOR Market Model (LMM), which specifies the dynamics of forward LIBOR rates associated with a fixed tenor structure $\{T_i\}$. Let $L_i(t)=L(t;T_i,T_{i+1})$ denote the forward rate for the accrual period $[T_i,T_{i+1}]$. The natural numeraire for $L_i$ is the zero-coupon bond $P(t,T_{i+1})$, and under the corresponding forward measure $\mathbb{Q}^{T_{i+1}}$ the forward rate is assumed to be a martingale. Its dynamics are
\begin{equation}
 \frac{dL_i(t)}{L_i(t)} = \sigma_i(t)^{\top} dW_t^{\mathbb{Q}^{T_{i+1}}}, \qquad t \le T_i.
\end{equation}

We impose a structured volatility assumption consistent with market practice and the objectives of this work. Specifically, each $\sigma_i(t)$ is assumed to be generated from a collection of time-dependent scalar functions multiplied by rate-specific constant proportionality factors. On any fixed time slab, this implies that the instantaneous covariance matrix of the forward rates is proportional to a constant positive-definite matrix. Such constructions are widely used to ensure stable calibration and are particularly well suited for orthogonal transformations that diagonalize the diffusion operator~\cite{BraceGM1997,BrigoM2006,AndersenP2010}.

When pricing products depending on multiple forward rates, such as swaptions, it is necessary to move to a single pricing numeraire, for example the domestic money-market account or a swap annuity. Under this change of numeraire the forward rates acquire drift terms that are explicitly state dependent, involving rational functions of other forward rates. In line with standard market approximations, we replace these state-dependent drifts by deterministic, time-dependent functions obtained through frozen-drift or piecewise-frozen approximations. After this modification, the resulting dynamics retain their diffusion structure with drifts depending only on time, enabling the associated pricing PDEs to be transformed into heat-type equations with time-dependent coefficients.

\subsubsection{Hybrid FX--interest-rate derivatives}

Hybrid FX--interest-rate derivatives are modeled by combining domestic and foreign LIBOR Market Models with an explicit diffusion model for the FX rate. Let $S_t$ denote the FX rate quoted in domestic currency per unit of foreign currency, and let $B_t^d$ and $B_t^f$ denote the domestic and foreign money-market accounts. Under the domestic risk-neutral measure $\mathbb{Q}^d$, associated with the numeraire $B_t^d$, the process $S_t B_t^f / B_t^d$ is assumed to be a martingale. The FX dynamics therefore take the form
\begin{equation}
 \frac{dS_t}{S_t} = (r_d(t)-r_f(t))\,dt + \sigma_S(t)^{\top} dW_t^{\mathbb{Q}^d}.
\end{equation}

Domestic and foreign forward rates are modeled under their respective forward measures, each being driftless under its own bond numeraire. As in the pure interest-rate case, we assume that volatilities of all domestic rates, foreign rates, and FX are generated from common time-dependent functions with asset-specific constant scaling factors and a fixed correlation structure. Consequently, the joint covariance matrix of all risk factors is time dependent but spatially constant on each time slab~\cite{Schlogl2002,Clark2011}.

Pricing hybrid derivatives requires a change to a single pricing numeraire, typically the domestic money-market account or a domestic annuity. This introduces drift terms coupling FX and interest-rate factors. As before, we replace these state-dependent drifts by deterministic, time-dependent approximations. Under these assumptions, the resulting pricing PDEs admit a reduction to heat-type equations, forming the basis for the quantum algorithms developed in this work.

A key structural distinction in hybrid products is between mark-to-market (MTM) and non-MTM contracts. In non-MTM derivatives, notionals are fixed in their respective currencies and FX exposure enters primarily through valuation. In MTM derivatives, notionals are periodically reset using the prevailing FX rate, leading to stronger coupling between FX dynamics and future cash flows. While both structures can be accommodated within the same modeling framework, MTM products typically exhibit stronger cross-asset interactions and place tighter requirements on the drift approximations used to obtain heat-type PDEs.

\section{Overview of our approach}
\label{sec:overview}

Now we provide a high-level overview of our approach. Let us consider a $D$-dimensional heat-type PDE of the form 
\begin{equation}
\frac{\partial u(x,t)}{\partial t}
=
\sum_{d=1}^D a_d(t)\, \frac{\partial^2 u(x,t)}{\partial x_d^2}
\label{eq:heat-pde}
\end{equation}
with a prescribed initial condition $u(x,0)=u_0(x)$. The first step is to convert the above PDE into a finite-dimensional ODE. To obtain such a representation, we discretize the spatial domain using a uniform grid with $N$ points per dimension and adopt a Fourier-based spectral discretization of the spatial differential operators. The discretization of the spatial domains incur errors that are well understood in the numerical analysis literature. Following the discretization, we develop a quanutm algorithm for solving the resulting ODE. This step will also incur errors due to approximations and simplifications we make. 

\paragraph{From parabolic PDEs to finite-dimensional ODEs}
To obtain a finite-dimensional representation, we discretize the spatial domain using a uniform grid with $N$ points per dimension. Rather than employing finite-difference method, we adopt a Fourier-based spectral discretization of the spatial differential operators. This approach is standard in the classical literature on numerical solutions of parabolic PDEs and is particularly well suited to problems with smooth solutions and periodic boundary conditions. In this discretization, each second order derivative operator with respect to $d$ is replaced by a discrete operator $\Delta^2_{x_d}$ defined by conjugation with the discrete Fourier transform such that the resulting operators are Hermitian and simultaneously diagonalizable in the Fourier basis.

Let $\mathbf{u}(t)\in\mathbb{C}^{N^D}$ denote the vector of grid-sampled values of $u(\cdot,t)$. The PDE \eqref{eq:heat-pde} is approximated by the linear ODE
\begin{equation}
\frac{d}{dt}\mathbf{u}(t)
=
A(t)\,\mathbf{u}(t),
\qquad
A(t):=\sum_{d=1}^D a_d(t)\,\Delta_{x_d}^2.
\label{eq:ode}
\end{equation}
Because the operators $\Delta_{x_d}^2$ commute, the generators $A(t)$ commute at different times, and the solution can be expressed as a matrix exponential involving time-integrated coefficients. This reduction is well known and provides a very good approximation to the original continuous problem with three different sources of errors introduced.

\begin{mylist}{\parindent}
    \item [1.] {\bf Spectral truncation error.} This error is introduced when we truncate the Fourier representation of a sufficiently smooth function by retaining only finitely-many modes depending on the grid resolution $N$. For solutions with bounded weak derivatives, this error decays algebraically with $N$ and can be made arbitrarily small by increasing the grid resolution.  

    \item [2.] {\bf Aliasing error.} This error is introduced  when the action of the differential operators generates high-frequency components that cannot be resolved on the discrete grid. Classical spectral methods control this effect by restricting the effective bandwidth of the representation and by exploiting the decay of Fourier coefficients for smooth solutions.

    \item [3.] {\bf Riemann-sum error.} The discrete Fourier transform approximates continuous Fourier integrals by finite sums. This approximation introduces a quadrature error that depends on the smoothness of the Fourier transform and the grid spacing. Under standard regularity assumptions, this error is smaller compared to truncation and aliasing errors, and decreases with increasing the gird resolution.
\end{mylist}

The combined discretization error from these sources quantifies the accuracy with which the finite-dimensional ODE approximates the original PDE. Importantly, these errors are well understood in the classical setting, and our choice of discretization follows established spectral methods designed to control them.

\paragraph{Quantum algorithm for the discretized ODE}
Once the discretized operator $A(t)$ is fixed, the remaining task is to compute the action of the associated propagator on the initial data. Assuming that the discretized vector $\mathbf{u}(t)$ into the amplitudes of a quantum state, the solution at the final time can be expressed in the form
\[
\mathbf{u}(T)=\exp\!\left(\int_0^T A(t)\,dt\right)\mathbf{u}(0).
\]
For notational simplicity, write the time-integrated generator as $A\succeq 0$ and the propagator as $e^{-\beta A}$ for an appropriate $\beta>0$ after rescaling. The quantum algorithm proceeds in three conceptual steps: (i) encode $\mathbf{u}(0)$ into a quantum state, (ii) implement a polynomial/LCU approximation of $e^{-\beta A}$, and (iii) extract the desired information from the resulting quantum state.

\begin{mylist}{\parindent}
    \item [1.] {\bf State preparation.} We encode the discretized initial condition into amplitudes of a $D\lceil\log(N)\rceil$-qubit state
\[
|u_0\rangle
=
\frac{1}{\|\mathbf{u}(0)\|}
\sum_{\mathbf{j}} \mathbf{u}_{\mathbf{j}}(0)\,\ket{\mathbf{j}}.
\]
One should view $\ket{\mathbf{j}}$ as $D$ quantum registers each of $\lceil\log (N)\rceil$ qubits. The overall complexity assumes that this state can be prepared efficiently. The error due to state-preparation error is explicitly accounted for as one of the sources of algorithmic error.

\item [2.] {\bf Implementing the propagator.} 
After normalizing $A$ so that its spectrum lies in a known interval (e.g., $[0,P]$), the exponential $e^{-\beta A}$ is approximated by a polynomial or Fourier-type expansion. A convenient representation is a Fourier-based LCU expansion
\[
e^{-\beta A}\approx \sum_{\ell=-M}^{M} c_\ell\,U^\ell
\qquad \text{where} \qquad
U:=e^{-i\pi A/P}.
\]
Here $M$ controls the truncation length and $\{c_\ell\}$ are computable coefficients. This representation can be implemented using an LCU circuit as shown in Figure~\ref{fig:Fourier_LCU}.

\item [3.] {\bf Extracting the solution.} Once the desired solution is prepared in a quantum state, we employ~\cite{GumaroS2026} to extract the solution surface from the state.
\end{mylist}

We are focused on the first two steps, and they introduce several distinct sources of errors. These error sources are algorithmic in nature and can be controlled independently of the spatial discretization by adjusting circuit parameters. These are explicity discussed in Section~\ref{sec:FF_LCU}.

\paragraph{Fast-forwardability.}
A central reason our quantum approach can be highly-efficient (in space and time complexity) is that the discretized diffusion operator admits a restricted form of \emph{fast-forwardability}. By construction, $A$ is a linear combination of commuting Fourier-diagonal operators:
\[
A=\sum_{d=1}^D \alpha_d\,\Delta_{x_d}^2
\qquad \text{such that} \qquad
\left[\Delta_{x_{d_1}}^2,\Delta_{x_{d_2}}^2\right]=0.
\]
Hence, $A$ is diagonalizable in the Fourier basis. Consequently, implementing the associated unitary
\[
U=e^{-i\pi A/P}
\]
is highly-efficient, and, at the same time, an approximation of
\begin{align}
    e^{-\beta A}
\end{align}
is efficiently implementable by Fourier-LCU method. This is explained comprehensively in Section~\ref{sec:FF_LCU}.

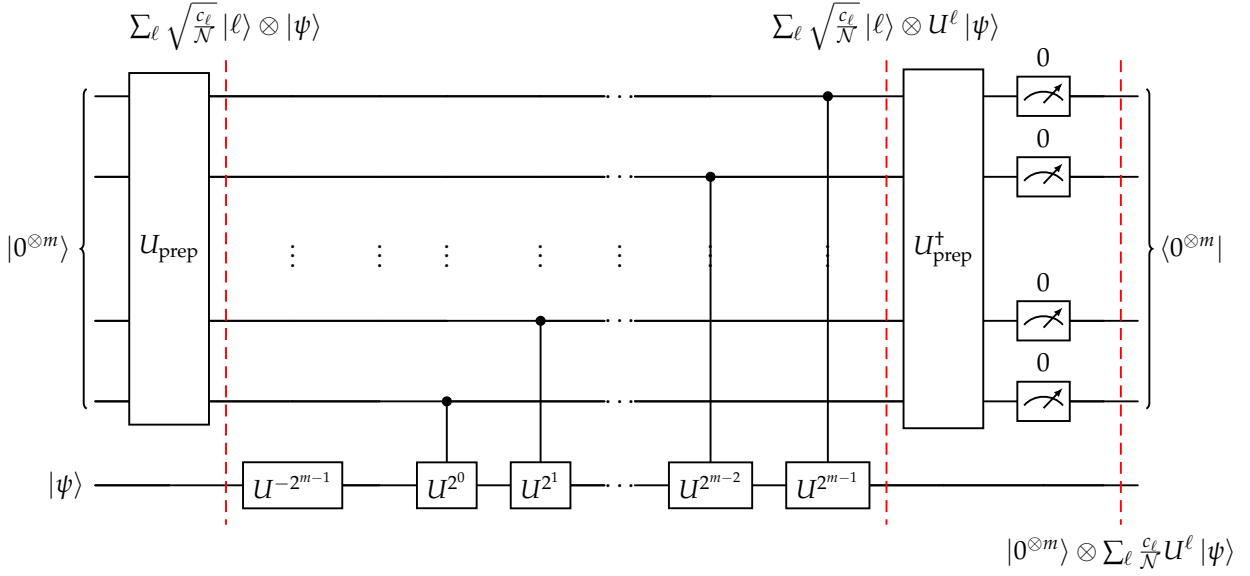
\begin{figure*}[!htb]
\centering
\scalebox{0.9}{
\begin{quantikz}
\lstick[wires=5]{$\ket{0^{\tp  m}}$} & \gate[wires=5,nwires=3]{U_{\rm prep}}\slice{$\sum_\ell \sqrt{\frac{{c}_{\ell}}{\mathcal{N}}}\ket{\ell} \tp \ket{\psi}$}& \qw &  \qw & \qw & \qw & \qw\ldots &\qw & \ctrl{5}\slice{$\sum_{\ell} \sqrt{\frac{{c}_{\ell}}{\mathcal{N}}}\ket{\ell} \tp U^\ell \ket{\psi}$}  & \gate[wires=5,nwires=3]{U_{\rm prep}^{\dagger}} & \meter{0} & \qw & \rstick[wires=5]{$\bra{0^{\tp m}}$}
 \\
   &\qw& \qw &\qw& \qw&\qw & \qw\ldots & \ctrl{4}   & \qw & \qw & \meter{0} & \qw &
 \\
  & \setwiretype{n}\vdots &\vdots& \vdots & \vdots & \vdots &\vdots&\vdots&\vdots & \vdots
 \\
  &\qw & \qw &\qw& \qw&\ctrl{2}  & \qw\ldots & \qw   & \qw & \qw & \meter{0} & \qw &
 \\
  & \qw & \qw &\qw& \ctrl{1}& \qw  & \qw\ldots & \qw  & \qw  & \qw & \meter{0} & \qw &
 \\
 \lstick{$\ket{\psi}$} &\qw & \gate{U^{-2^{m-1}}} & \qw & \gate{U^{2^0}}& \gate{U^{2^1}} & \qw\ldots & \gate{U^{2^{m-2}}} & \gate{U^{2^{m-1}}} & \qw & \qw & \qw \slice[label style={pos=1, anchor=north}]{$\ket{0^{\tp m}} \tp \sum_{\ell} \frac{{c}_{\ell}}{\mathcal{N}}U^\ell \ket{\psi} $} & 
\end{quantikz}
}
\caption{Circuit to implement an $m$-qubit LCU implementation of $e^{-\beta A}$. $U$ is $e^{-i \pi  A/P}$. The $U_{\rm prep}$ applied on $\ket{0}^{\tp m}$ prepares the square root of the positive-definite coefficients of the LCU. }\label{fig:Fourier_LCU}
\end{figure*}

\section{Fast Fourier Linear Combination of Unitaries for $e^{-\beta A}$}
\label{sec:FF_LCU}

Let $A$ be a positive semidefinite, fast-forwardable operator with $\|A\|\leq P$. Our goal is to implement the non-unitary operator $e^{-\beta A}$.

\paragraph{Fourier representation.}
Define
\begin{align}
U := e^{-i\pi A / P}.
\end{align}
Then $e^{-\beta A}$ admits the Fourier series representation
\begin{align}
e^{-\beta A}
=
\sum_{\ell\in\mathbb{Z}} c_\ell\, U^\ell.
\end{align}

In practice, we truncate to
\begin{align}
e^{-\beta A}
\approx
\sum_{\ell=-M/2}^{M/2-1} c_\ell\, U^\ell,
\qquad M = 2^m.
\end{align}

\paragraph{LCU implementation.}
This truncated expansion can be implemented using a Linear Combination of Unitaries (LCU) circuit (see \Cref{fig:Fourier_LCU}). The circuit prepares
\begin{align}
U_{\rm prep}\ket{0^{\otimes m}}
=
\sum_{\ell=-M/2}^{M/2-1}
\sqrt{\frac{c_\ell}{\mathcal{N}}}\,
\ket{\ell + M/2},
\end{align}
followed by controlled powers of $U$.

Because the series in the LCU consists only of powers of $U$ and $U$ is fast-forwardable, the sequence of controlled operations requires only
\begin{align}\label{eq:}
\mathcal{O}(m) = \mathcal{O}(\log M)
\end{align}
applications of $U^{2^j}$ where each is efficiently implementable. This brings the circuit depth to $\mathcal{O}(m)$ rather than $\mathcal{O}(M)$.

The bottleneck is therefore the preparation of $U_{\rm prep}$, which in general costs $\mathcal{O}(M)$.

\paragraph{Efficient coefficient preparation.}
To overcome this, we approximate the coefficients $c_\ell$ using a distribution amenable to Grover--Rudolph state preparation:
\begin{align}
\tilde{U}_{\rm prep}\ket{0^{\otimes m}}
=
\sum_{\ell=-M/2}^{M/2-1}
\sqrt{p_\ell}\,\ket{\ell+M/2},
\end{align}
where
\begin{align}
p_\ell = \int_\ell^{\ell+1} p(x)\,{\rm d}x,
\qquad
p(x)
=
\frac{1}{\mathcal{N}_2}
\frac{\beta P}{(\beta P)^2 + (\pi x)^2}.
\end{align}

---

\subsection{Exact Fourier Coefficients}

We compute the Fourier coefficients of
\begin{align}
f(x) = e^{-\beta |x|}, \quad x \in [-P,P].
\end{align}

A direct calculation yields
\begin{align}
c_\ell
=
\frac{\beta P\bigl(1 - (-1)^\ell e^{-\beta P}\bigr)}
{(\beta P)^2 + \pi^2 \ell^2}.
\label{eq:fourier_coeff_exact}
\end{align}

---

\subsection{Approximation of Coefficients}

\paragraph{First approximation.}
Assuming $\beta P \gg 1$, we neglect the exponentially small term and define
\begin{align}\label{eq:c_tilde}
\tilde{c}_\ell
=
\frac{\beta P}{(\beta P)^2 + \pi^2 \ell^2}.
\end{align}

The induced operator error satisfies
\begin{align}
\left\|
\sum_\ell c_\ell U^\ell
-
\sum_\ell \tilde{c}_\ell U^\ell
\right\|
=
\mathcal{O}(e^{-\beta P}).
\end{align}

---

\subsection{Truncation Error}

Truncating the series to $|\ell|<M/2$ yields
\begin{align}
\left\|
\sum_{\ell\in\mathbb{Z}} \tilde{c}_\ell U^\ell
-
\sum_{\ell=-M/2}^{M/2-1} \tilde{c}_\ell U^\ell
\right\|
\le
2\sum_{\ell\ge M} \tilde{c}_\ell.
\end{align}

Using an integral bound,
\begin{align}
\sum_{\ell\ge M} \tilde{c}_\ell
\le
\int_{M-1}^\infty
\frac{\beta P}{(\beta P)^2 + \pi^2 x^2}\,dx
=
\frac{1}{\pi}
\arctan\!\left(\frac{\beta P}{\pi(M-1)}\right),
\end{align}
which implies
\begin{align}
\text{Truncation error}
=
\mathcal{O}\!\left(\frac{\beta P}{M}\right).
\end{align}

---

\subsection{Continuous Approximation Error}

To enable efficient state preparation, we approximate
\begin{align}
\tilde{c}_\ell \approx \tilde{p}_\ell := \int_\ell^{\ell+1} \tilde{p}(x)\,dx,
\quad
\tilde{p}(x) = \frac{\beta P}{(\beta P)^2 + \pi^2 x^2}.
\end{align}

Using bounds on the derivative of $\tilde{p}(x)$, we obtain
\begin{align}
|\tilde{c}_\ell - \tilde{p}_\ell|
=
\mathcal{O}\!\left(\frac{1}{\beta^2 P^2}\right).
\end{align}

Thus,
\begin{align}
\left\|
\sum_{\ell} \tilde{c}_\ell U^\ell
-
\sum_{\ell} \tilde{p}_\ell U^\ell
\right\|
=
\mathcal{O}\!\left(\frac{M}{\beta^2 P^2}\right).
\end{align}

---

\subsection{Renormalization Error}

The normalization constant is
\begin{align}
\mathcal{N}_2
=
\int_{-M/2}^{M/2}
\frac{\beta P}{(\beta P)^2 + \pi^2 x^2}\,dx
=
\frac{2}{\pi}
\arctan\!\left(\frac{\pi M}{2\beta P}\right).
\end{align}

For large $M$, this gives
\begin{align}
\mathcal{N}_2 = 1 - \mathcal{O}\!\left(\frac{\beta P}{M}\right).
\end{align}

Hence,
\begin{align}
\text{Renormalization error}
=
\mathcal{O}\!\left(\frac{\beta P}{M}\right).
\end{align}

---

\subsection{Total Error and Parameter Matching}

Combining all contributions, the total operator error is
\begin{align}
\mathcal{O}\!\left(
e^{-\beta P}
+
\frac{M}{\beta^2 P^2}
+
\frac{\beta P}{M}
\right).
\end{align}

The exponential term is negligible for $\beta P \gg 1$. Balancing the remaining terms yields
\begin{align}
\beta P \sim M^{2/3},
\end{align}
which gives overall scaling
\begin{align}
\text{Total error}
=
\mathcal{O}(M^{-1/3}).
\label{eq:FFLCU_error_scaling}
\end{align}

\section{Discretization of Diffusion and Drift Operators}
\label{sec:discretization_errors}

We describe the discretization of the diffusion operator $e^{\kappa \partial_x^2}$ and the drift operator $e^{\gamma \partial_x}$, with emphasis on constructions that are \emph{fast-forwardable} and compatible with efficient quantum implementation.

\paragraph{Assumptions.}
Let $\psi:\mathbb{R}\to\mathbb{C}$ be such that its weak second derivative $\psi^{(2)}$ is integrable. We further assume compact support up to negligible leakage:
\begin{align}
\psi(x) \approx 0 \quad \text{for } |x| > L,
\end{align}
and after evolution,
\begin{align}
\psi_{\rm evolved}(x) \approx 0 \quad \text{for } |x| > L',
\end{align}
up to error $\epsilon_{\rm leak}$, where
\begin{align}
L' = \mathcal{O}\!\left( L + \kappa\,\polylog{1/\epsilon_{\rm leak}} + \gamma \right).
\end{align}

\paragraph{Fourier decay.}
A standard result implies polynomial decay of Fourier coefficients:

\begin{theorem}[Fourier decay]\label{thm:fourier_decay}
If $f \in L^1(\mathbb{R})$ and its weak derivative $f^{(m)} \in L^1(\mathbb{R})$, then its Fourier transform satisfies
\begin{align}
|\widehat{f}(k)| = \mathcal{O}(|k|^{-m}) \quad \text{as } |k| \to \infty.
\end{align}
\end{theorem}

Applying this with $m=2$, we obtain
\begin{align}
|\Psi(f)| = \mathcal{O}\!\left( \frac{1}{f^2} \right).
\end{align}

\paragraph{Aliasing error.}
Approximating $\Psi(f)$ with the discrete-time Fourier transform (DTFT),
\begin{align}
\Psi_{1/\delta}(f)
= \delta \sum_{n\in\mathbb{Z}} \psi(x_n)\, e^{-i2\pi f x_n},
\quad x_n = n\delta,
\end{align}
introduces aliasing error
\begin{align}
\left|\Psi(f) - \Psi_{1/\delta}(f)\right|
\le
\sum_{j\neq 0} |\Psi(f - j/\delta)|
=
\mathcal{O}(\delta^2).
\end{align}

---

\subsection{Diffusion Operator}

We consider
\begin{align}
e^{\kappa \partial_x^2}\psi(x)
=
\int_{\mathbb{R}}
e^{i2\pi f x}\,
e^{-\kappa (2\pi)^2 f^2}\,
\Psi(f)\,{\rm d}f.
\end{align}

\paragraph{Spectral truncation.}
Restricting to $f \in [-L_f,L_f]$, where $L_f = 1/(2\delta)$, yields truncation error
\begin{align}
\mathcal{O}\!\left(
\frac{e^{-\kappa (2\pi L_f)^2}}{L_f}
\right).
\end{align}

\paragraph{Riemann-sum approximation.}
Approximating the truncated integral with $N$ nodes,
\begin{align}
\sum_{k=-N/2}^{N/2-1}
e^{i2\pi f_k x}
e^{-\kappa (2\pi f_k)^2}
\Psi(f_k)\,\delta_f,
\quad f_k = k\delta_f,
\end{align}
with $\delta_f N/2 = L_f$, incurs error
\begin{align}
\mathcal{O}\!\left(\frac{L_f^3}{N}\right),
\end{align}
assuming $L_f \gg L'$.

\paragraph{Balancing scales.}
Choosing
\begin{align}
L_f = \Theta(N^{1/6}),
\end{align}
we obtain
\begin{align}
\text{Riemann error} &= \mathcal{O}(N^{-1/2}), \\
\text{Truncation error} &= \mathcal{O}\!\left(
\frac{e^{-\kappa (2\pi)^2 N^{1/3}}}{N^{1/3}}
\right), \\
\text{Aliasing error} &= \mathcal{O}(N^{-1/6}).
\end{align}

Thus, the dominant error scales as $\mathcal{O}(N^{-1/6})$.

\paragraph{Quantum implementation.}
We encode
\begin{align}
\ket{\psi}
=
\frac{1}{\sqrt{\mathcal{N}}}
\sum_{j=-N/2}^{N/2-1}
\psi(x_j)\ket{j+N/2},
\quad x_j = j\delta.
\end{align}

Define the DTFT operator
\begin{align}
F_{\rm DTFT}
=
\frac{1}{\sqrt{N}}
\sum_{k,j=-N/2}^{N/2-1}
e^{-i2\pi f_k x_j}
\ket{k+N/2}\bra{j+N/2}.
\end{align}

This operator is unitary iff
\begin{align}
\delta\,\delta_f = \frac{1}{N}.
\end{align}

For $N=2^n$, one obtains
\begin{align}
F_{\rm DTFT}
=
\hat{Z}_0\, \operatorname{QFT}^\dagger \hat{Z}_0,
\end{align}
up to a global phase.

\paragraph{Diagonal evolution.}
In Fourier space, the diffusion operator is diagonal:
\begin{align}
e^{-\kappa (2\pi)^2 f^2}
\;\longrightarrow\;
e^{-\kappa A},
\end{align}
with
\begin{align}
A = (2\pi)^2 \delta_f^2
\sum_k k^2 \ket{k+N/2}\bra{k+N/2}.
\end{align}

We approximate $e^{-\kappa A}$ using Fourier LCU:
\begin{align}
e^{-\kappa A}
\approx
\sum_{\ell=-M/2}^{M/2-1}
c_\ell\, e^{-i\pi \ell A / P}.
\end{align}

Using
\begin{align}
\sum_k k \ket{k+N/2}\bra{k+N/2}
=
-\tfrac{1}{2} - \sum_{b=0}^{n-1} 2^{b-1}\hat{Z}_b,
\end{align}
we obtain a fast-forwardable implementation of $e^{-i\pi A/P}$.

\paragraph{Block-encoding.}
The full non-unitary operator is
\begin{align}
U_{e^{\kappa \partial_x^2}}
=
F_{\rm DTFT}^\dagger
\,U_{e^{-\kappa A}}\,
F_{\rm DTFT}.
\end{align}

\begin{remark}
Choosing $P \ge \|A\|$ gives
\begin{align}
P \ge (2\pi)^2 \delta_f^2 (N/2)^2,
\end{align}
which implies
\begin{align}\label{eq:M_vs_N_diffusion}
M = \Omega\!\left( \kappa^{3/2} \sqrt{N} \right).
\end{align}
which is a further constraint from those derived for $M$ and $N$ with respect to $\epsilon$.
\end{remark}

---

\subsection{Drift Operator}

We now consider
\begin{align}
e^{\gamma \partial_x}\psi(x)
=
\int e^{i2\pi f x}
e^{-i2\pi \gamma f}
\Psi(f)\,{\rm d}f.
\end{align}

\paragraph{Truncation.}
Restricting to $[-L_f,L_f]$ yields error
\begin{align}
\mathcal{O}(L_f^{-1}).
\end{align}

\paragraph{Riemann approximation.}
The discretized sum
\begin{align}
\sum_{k=-N/2}^{N/2-1}
e^{i2\pi f_k x}
e^{-i2\pi \gamma f_k}
\Psi(f_k)\,\delta_f
\end{align}
has error
\begin{align}
\mathcal{O}\!\left(\frac{L_f^2}{N}\right).
\end{align}

\paragraph{Scaling.}
With $L_f = \Theta(N^{1/6})$, we obtain
\begin{align}
\text{Truncation} &= \mathcal{O}(N^{-1/6}), \\
\text{Aliasing} &= \mathcal{O}(N^{-1/6}), \\
\text{Riemann error} &= \mathcal{O}(N^{-2/3}).
\end{align}

Thus, the dominant error is again $\mathcal{O}(N^{-1/6})$.

\paragraph{Implementation.}
The circuit construction follows the same structure as for the diffusion operator, replacing the quadratic phase with a linear phase in Fourier space.

\section{Overall Resource Requirements}\label{sec:overall_resource}

In this work we consider the $D$-dimensional heat equation in the form
\begin{align}
\partial_t u(t,x)
&= \frac{1}{2}\sum_{i=1}^D \sigma_i^2(t)\,\partial_{x_i}^2 u(t,x)
+ \sum_{i=1}^D \gamma_i(t)\,\partial_{x_i} u(t,x),
\qquad
u(T,x) = \Phi_T(x),
\end{align}
where $x = (x_1,\dots,x_D)$.

\paragraph{Time reversal.}
Introduce the time-to-maturity variable $\tau = T - t$, and define
\begin{align}
w(\tau,x) := u(T-\tau,x).
\end{align}
Then $w$ satisfies
\begin{align}
\partial_\tau w(\tau,x)
&= \frac{1}{2}\sum_{i=1}^D \sigma_i^2(T-\tau)\,\partial_{x_i}^2 w(\tau,x)
- \sum_{i=1}^D \gamma_i(T-\tau)\,\partial_{x_i} w(\tau,x),
\qquad
w(0,x) = \Phi_T(x).
\end{align}

\paragraph{Operator solution.}
The formal solution can be written as
\begin{align}
w(\tau,x)
=
\exp\!\left(
\frac{1}{2}\sum_{i=1}^D
\left[
\int_0^\tau \sigma_i^2(T-s)\,{\rm d}s
\right]
\partial_{x_i}^2
\right)
\exp\!\left(
-
\sum_{i=1}^D
\left[
\int_0^\tau \gamma_i(T-s)\,{\rm d}s
\right]
\partial_{x_i}
\right)
w(0,x).
\end{align}

The drift term corresponds to a translation operator. Since it commutes with the diffusion operator, we obtain
\begin{align}
w(\tau,x)
=
\exp\!\left(
\frac{1}{2}\sum_{i=1}^D
\left[
\int_0^\tau \sigma_i^2(T-s)\,{\rm d}s
\right]
\partial_{x_i}^2
\right)
w\bigl(0,\,x - \Gamma(\tau)\bigr),
\end{align}
where $\Gamma(\tau) = (\Gamma_1(\tau),\dots,\Gamma_D(\tau))$ with
\begin{align}
\Gamma_i(\tau)
=
\int_0^\tau \gamma_i(T-s)\,{\rm d}s.
\end{align}

\paragraph{Discrete Laplacian approximation.}
As discussed in \Cref{sec:discretization_errors}, we replace the continuous operators $\partial_{x_j}^2$ with discrete, spectrally bounded operators $\Delta_{x_j}^2$ implemented via quantum Fourier transforms.

The discretized solution state is then
\begin{align}
\ket{u}
=
\prod_{j=1}^D
\exp\!\bigl( \beta\, \Delta_{x_j}^2 \bigr)
\ket{u_0},
\end{align}
where $\ket{u_0}$ encodes the shifted initial condition.

The discretization error satisfies
\begin{align}
\left\Vert \ket{w} - \ket{u} \right\Vert_{\infty}
=
\mathcal{O}\!\left( D\, N^{-1/6} \right),
\end{align}
where $(\ket{w})_{\vec{x}} = w(\vec{x})$.

\paragraph{Fourier LCU approximation.}
Using the Fourier LCU construction of \Cref{sec:FF_LCU}, we approximate $\ket{u}$ by
\begin{align}
\ket{v}
=
\prod_{j=1}^{D}
\left(
\sum_{\ell=-M/2}^{M/2-1}
p_\ell\, U_{x_j}^{\ell}
\right)
\ket{u_0},
\end{align}
with error
\begin{align}
\left\Vert \ket{u} - \ket{v} \right\Vert_{\infty}
=
\mathcal{O}\!\left( D\, M^{-1/3} \right).
\end{align}

\paragraph{Parameter scaling.}
To achieve total error $\epsilon$, it suffices to take
\begin{align}
M = \Omega\!\left( \frac{D^3}{\epsilon^3} \right),
\qquad
N = \Omega\!\left( \frac{D^6}{\epsilon^6} \right).
\end{align}

\paragraph{Block-encoding.}
We construct a unitary $U_{\ket{v}}$ that block-encodes $\ket{v}$:
\begin{align}
U_{\ket{v}}
=
\frac{1}{\mathcal{N}}
\ket{0^{\otimes m}}\!\bra{0^{\otimes m}}
\otimes
\ket{v}\!\bra{v}
+ G,
\end{align}
where $G$ is a garbage operator satisfying
\begin{align}
G \ket{0^{\otimes m}} \otimes \ket{\psi} = 0
\quad \text{for all } \ket{\psi}.
\end{align}
The normalization satisfies
\begin{align}
\mathcal{N}
=
\sum_{\vec{j}} w_0^2(\vec{j}/N)
\approx N^D.
\end{align}

\paragraph{Complexity.}
The preparation circuit is described in \Cref{fig:algorithm} and \Cref{fig:Fourier_LCU}. As shown in \Cref{sec:FF_LCU}, its cost is
\begin{align}
\mathcal{O}\!\bigl( \polylog{1/\epsilon} \bigr).
\end{align}

Combining with \Cref{lem:SE_multidim}, the total cost of reconstructing the solution surface is
\begin{align}
\tilde{\mathcal{O}}
\!\left(
\frac{1}{\epsilon}
\, 2^{D/2}
\, L^{4D}
\right),
\end{align}
where
\begin{align}
L = \mathcal{O}\!\left( \Lambda + \frac{1}{D}\log\frac{1}{\epsilon} \right).
\end{align}

The smoothness parameter $\Lambda$ satisfies
\begin{align}
\left| \partial^\alpha w(\vec{x},T) \right|
\le
\Lambda^{|\alpha|+1}.
\end{align}
In \Cref{sec:gaussian_derivative_bounds}, we show that for heat equation solutions,
\begin{align}
\Lambda = \mathcal{O}\!\left( \frac{1}{\sqrt{2\kappa T}} \right),
\end{align}
up to an $\epsilon$-approximation error.

\begin{figure} \label{fig:algorithm}
\caption{Pseudo-code for Circuit Construction for $U_{\ket{v}}$}
\begin{mdframed}[linewidth=1pt, roundcorner=4pt]
\begin{algorithmic}[1]

\Require Inputs: $\beta$, $M$, $N$, $L$, $w_{0}(\mathbf{x})$
\Ensure Output: Unitary circuit, $U_{\ket{v}}$, block-encoding $w(\mathbf{x},T)$

\State On the target $D\times N$ qubits, add the circuitry needed to prepare the initial conditions $w_0(\mathbf{x})$. This is done efficiently through Grover-Rudolph method, provided the integral of $(w_0)^2(\mathbf{x})$ is efficiently computable with classical (reversible) arithmetic

\State Apply the preparation circuit, $U_{\rm prep}$, that prepares the approximate square root of the Fourier series coefficients on the $M$ ancillar qubits using Grover-Rudolph.

\State Perform the unitary $U^{-2^{m-1}}$, where $m=\log_2 M$

\For{$j \gets 0$ to $m-1$}
    \State Perform the unitary $U^{2^{j}}$ controlled on the $j$-th qubit
\EndFor

\State 

\State \Return Resulting circuit for $M + D\times N$ qubits

\end{algorithmic}
\end{mdframed}
\end{figure}

\section{Numerical simulations}

To demonstrate our technique, we carry out a small-scale numerical simulation using an implementation in \textsc{pytket}. Here we solve the one dimensional heat equation $\frac{\partial \psi}{\partial t} = \frac{\partial^2 \psi}{\partial x^2}$ on $(x,t) \in \mathbb{R}\times [0,\infty)$ with an initial condition $\psi(x,t=0) = \psi_0(x) \propto \cos(5\pi x) + 2\cos(\pi x) + 4$ (up to the $\ell_2$ normalization) and periodic boundary conditions. The exact solution is then obtained using the heat kernel Green's function as \begin{align}
    \psi(x,t) = \frac{1}{\sqrt{4\pi t}} \int_{-\infty}^\infty \psi_0(y) \cdot e^{-\frac{(x-y)^2}{4t}}\ \dd y\,.\label{eqn:heat equation example, exact soln}
\end{align}

We start by encoding the initial condition into a quantum state using the Grover-Rudolph method, to which we then apply the block encoding \Cref{fig:Fourier_LCU} of the time propagator and project it onto the appropriate subspace to carry out the LCU. In Figure \ref{fig:heat_equation}, we compare the measured amplitudes to the exact solution \eqref{eqn:heat equation example, exact soln} for different values of time $t$. On Subfigures \ref{fig:heat_equation_partA}-\ref{fig:heat_equation_partC}, even with only $14$ total qubits, we see a very accurate match between the evolved state and the exact solution. However, we require the number of ancilla qubits to depend logarithmically on the evolution time, as specified in \eqref{eq:M_vs_N_diffusion}, and so for a fixed number of qubits as in these simulations, we will see the state stop evolving after some maximal time. This is demonstrated on Subfigure \ref{fig:heat_equation_partD}, where we see the evolved state unchanged from that of Subfigure \ref{fig:heat_equation_partC}, while the exact solution kept on evolving. This is to be expected, as for large times, any fixed number of the LCU coefficients given by \eqref{eq:c_tilde} will become uniform (although still decreasing with $t$ in magnitude, which, however, won't be encoded in the quantum state representing the coefficients), and so the result of the LCU after this maximal time will stay the same.

\begin{figure}[H]
    \centering
    \subfloat[$t=0.001$]{
            \includegraphics[width=0.5\textwidth]{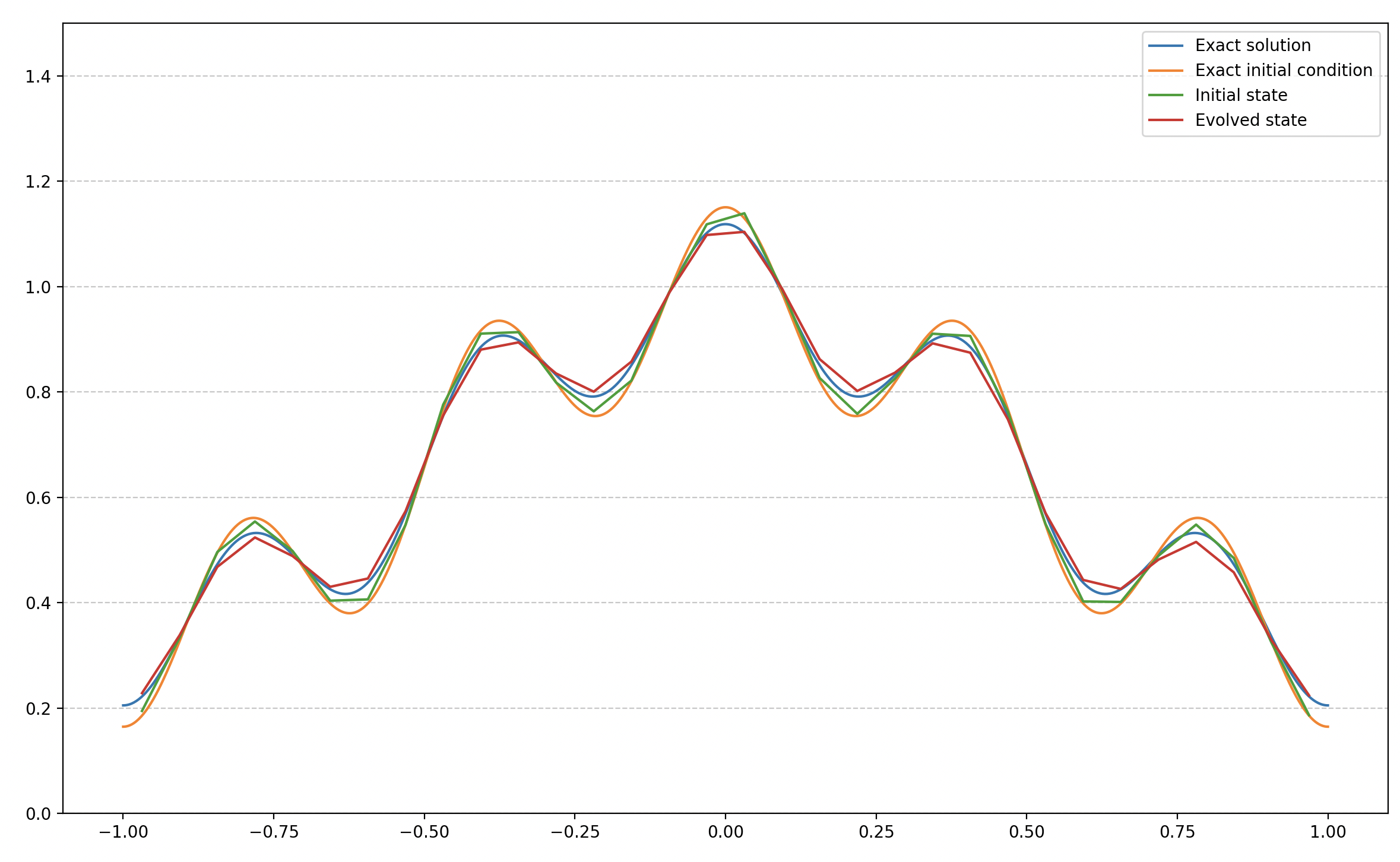}\label{fig:heat_equation_partA}}
    \subfloat[$t=0.0025$]{
            \includegraphics[width=0.5\textwidth]{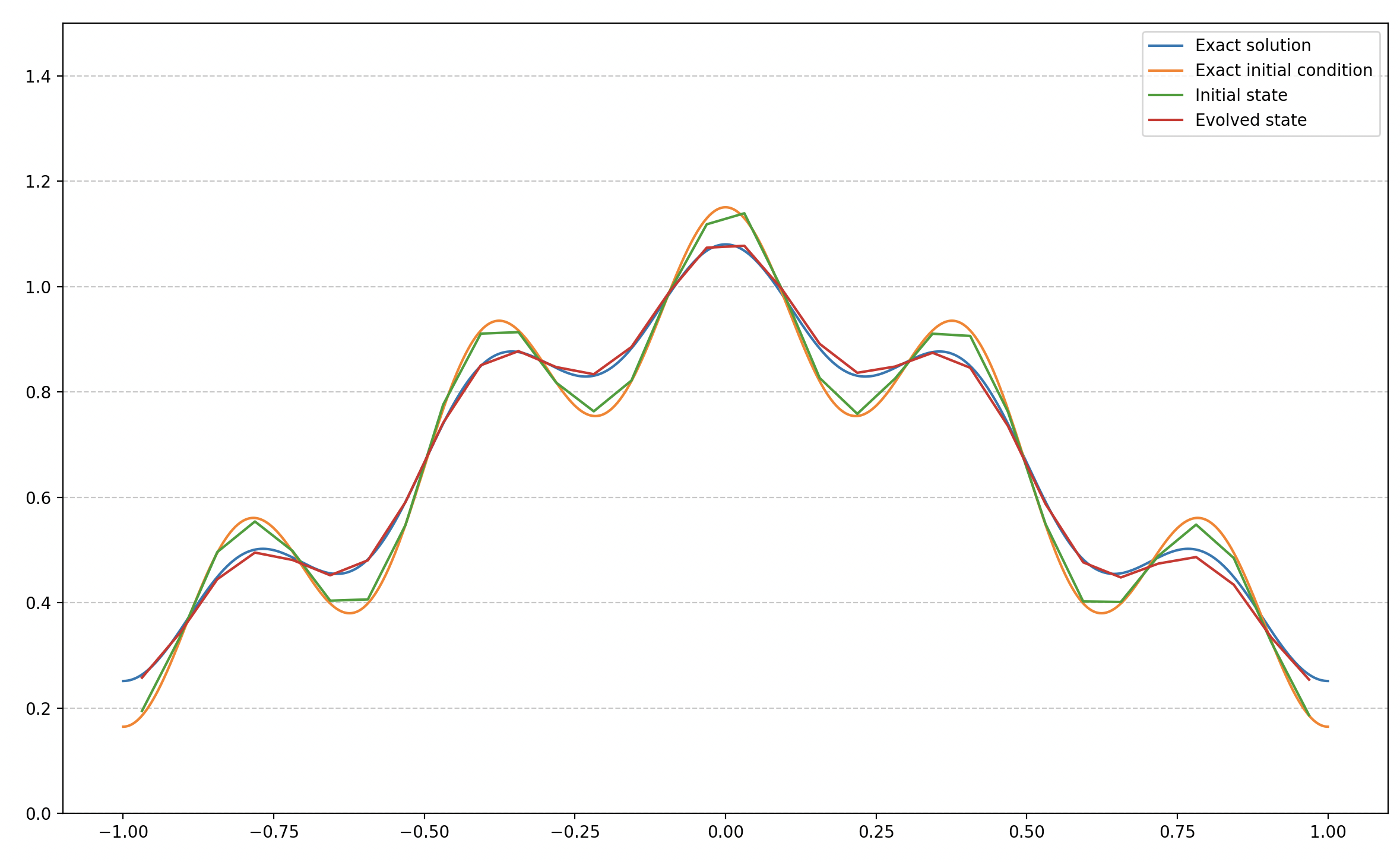}\label{fig:heat_equation_partB}}\\
    \subfloat[$t=0.005$]{
            \includegraphics[width=0.5\textwidth]{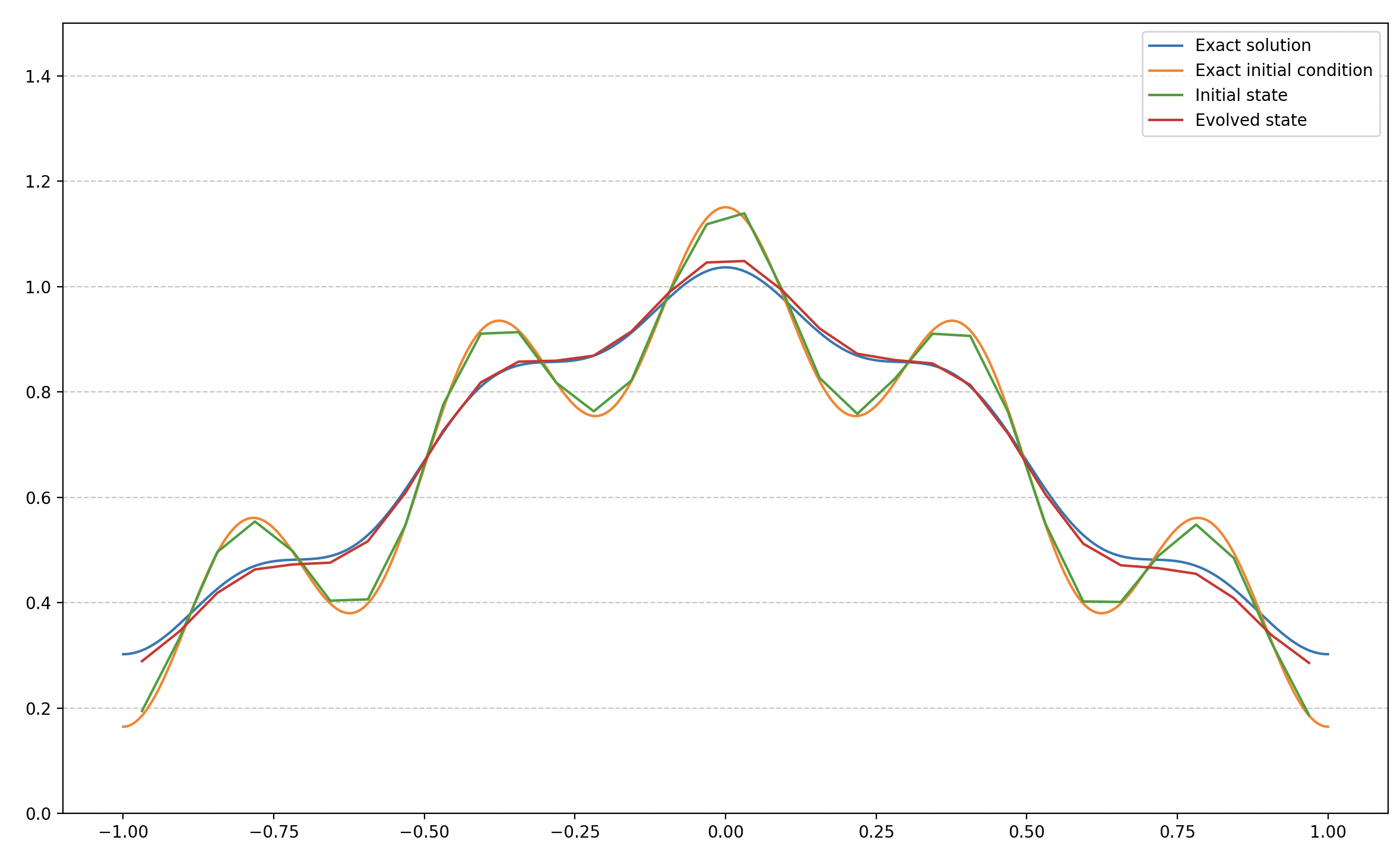}\label{fig:heat_equation_partC}}
    \subfloat[$t=0.15$]{
            \includegraphics[width=0.5\textwidth]{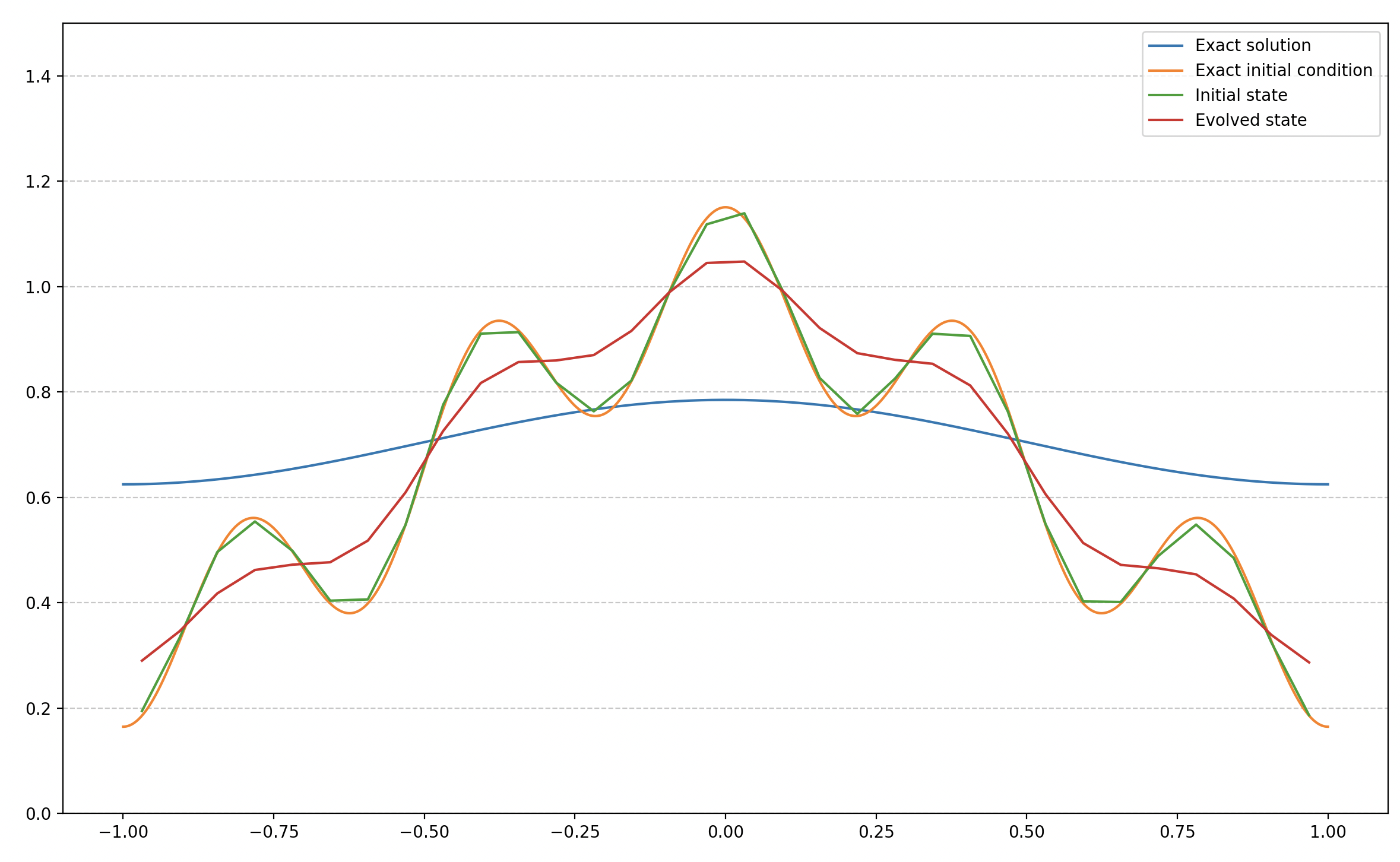}\label{fig:heat_equation_partD}}
    \caption{Solving the heat equation $\partial_t \psi = \partial^2_x\psi$ with initial condition $\psi_0(x) \propto \cos(5\pi x) + 2\cos(\pi x)+4$ and periodic boundary conditions, comparing the exact solution to the one obtained with FFLCU, using $n=5$ qubits for the state, $m=4$ qubits for the LCU, and $k = 5$ ancilla qubits for Grover-Rudolph method used for encoding the LCU coefficients. Here, all functions are always rescaled to have $\ell_2$ norm equal to $1$ (note that this property is normally not preserved during evolution under heat equation). Correspondingly, all measured amplitudes are rescaled by $\sqrt{2^{n-1}}$ to match.}
    \label{fig:heat_equation}
\end{figure}

\section{Applications}

The obvious application to come to mind is the study of heat transfer. Apart from this, we can apply this kind of methods to the pricing of financial contracts. In the following, we will explain how we can use these methods for pricing of multi-asset European option and Bermudan option contracts as well as pricing \textit{swaption} contracts.

\subsection{Basket Option Pricing}
We start with the problem of pricing a basket option under the Black-Scholes model, that is, its PDE formulation is:

\paragraph{PDE formulation.}
Let \(V:[0,T]\times(0,\infty)^D\to\R\) be sufficiently smooth. Then \(V\) solves
\[
\partial_t V
+\sum_{i=1}^D (r-q_i) S_i\,\partial_{S_i} V
+\frac{1}{2}\sum_{i=1}^D\sum_{j=1}^D
\rho_{ij}\sigma_i\sigma_j\, S_i S_j\, \partial_{S_i S_j}^2 V
- r V \;=\; 0,\qquad t<T,
\]
with terminal condition
\[
V(T,s)=g(s)=\bigl(w^\top s - K\bigr)^+,
\]
and suitable growth/absorbing boundary conditions as \(S_i\downarrow 0\) and \(S_i\uparrow\infty\).
The terminal condition greatly simplifies if we rotate the assets axes such that the vector $w$ points fully into any of these new rotate independent variables.

In order to implement the solution procedure we first switch to $\log$-variables and reverse time:
\begin{align}
    x_i &= \log{S_i} \cr
    \tau &= T-t \cr 
    U (\tau,x)&= V(t,s).
\end{align}
With this, we obtain:
\[
-\partial_{\tau} U
+\sum_{i=1}^D (r-q_i) \,\partial_{x_i} U
+\frac{1}{2}\sum_{i=1}^D\sum_{j=1}^D
\rho_{ij}\sigma_i\sigma_j\, \partial_{x_i x_j}^2 U
- r U \;=\; 0,\qquad  0 < \tau<T.
\]
Following this, we remove the drift terms with the moving frames
\begin{align}
    y_i = x_i + \tau (r-q_i)
\end{align}
we remove the zeroth order term with the definition
\begin{align}
    W(\tau,y) = e^{r \tau} U(\tau,x),
\end{align}
obtaining:
\begin{align}
    \partial_\tau W - \frac{1}{2}\sum_{i=1}^D\sum_{j=1}^D\Sigma_{ij} \partial_{y_i} \partial_{y_j} W = 0
\end{align}

We take the covariance matrix:
\begin{align}
    \Sigma_{ij} = \rho_{ij} \sigma_i \sigma_j.
\end{align}
We can diagonalize this positive-definite matrix through:
\begin{align}
    \Sigma = Q \Lambda Q^{\top}.  
\end{align}
Using this definition we establish a new change of variables:
\begin{align}
    z &= Q^{\top} y,\cr
    K(\tau,z) &= W(\tau,y).
\end{align}
Through this, we obtain:
\begin{align}
    \partial_\tau K - \frac{1}{2}\sum_{i=1}^D\Lambda_{ii} \partial^2_{z_i}  K = 0.
\end{align}
Finally, we make the heat equation anisotropic by rescaling each variable through
\begin{align}
    w&=\left(\Lambda\right)^{-1/2} z \cr
    L(\tau,w)&= W(\tau,z),
\end{align}
thus,
\begin{align}
    \partial_\tau L - \frac{1}{2}\sum_{i=1}^D \partial^2_{w_i}  L = 0.
\end{align}
The overall cost of our method to price $L$ is
\begin{align}
    \mathcal{O}\left( \frac{1}{\epsilon}  M^{4D} \right)
\end{align}
where $M$ is the number of interpolation samples per dimension, which scales the following way ~\Cref{lem:SE_multidim}:
\begin{align}
    \mathcal{O}\left(\Lambda + \frac{1}{D}\log{1/\epsilon} \right).
\end{align}
If one were to obtain the whole surface using Quantum Monte Carlo methods, the expected scaling is similar since one needs to sample $M^{D}$ interpolation points as well. This with a similar scaling in qubit requirements: $\mathcal{O}\left(\polylog{1/\epsilon}\right)$.

\subsubsection{Multi-Asset Bermudan Options}\label{subsec:Bermudan}

A Bermudan option allows early exercise only at a discrete set of times $0=t_0<t_1 < t_2 < \dots < t_N = T$. Between exercise dates, the value function$V(\mathbf{S}, t)$, with $\mathbf{S} = (S_1, S_2, \dots, S_D)$ the vector of underlying asset prices, satisfies the multi-dimensional Black--Scholes partial differential equation. For \( t \in (t_i, t_{i+1}) \), the value function \( V(\mathbf{S}, t) \) satisfies
\begin{align}
\frac{\partial V}{\partial t}
+ \frac{1}{2} \sum_{j=1}^{D} \sum_{k=1}^{D} 
\rho_{jk} \sigma_j \sigma_k S_j S_k 
\frac{\partial^2 V}{\partial S_j \partial S_k}
+ \sum_{j=1}^{D} (r - q_j) S_j \frac{\partial V}{\partial S_j}
- r V = 0,
\label{eq:multi_bs_pde}
\end{align}
where
\begin{itemize}
    \item \( r \) is the risk-free rate,
    \item \( q_j \) is the dividend yield of asset \( j \),
    \item \( \sigma_j \) is the volatility of asset \( j \),
    \item \( \rho_{jk} \) is the constant correlation between assets \( j \) and \( k \).
\end{itemize}

At maturity \( T = t_N \), the terminal condition is the payoff:
\begin{align}
V(\mathbf{S}, T) = \Phi(\mathbf{S}),
\end{align}
typically,
$$
\Phi(\mathbf{S}) = \max\!\big( \sum_{j=1}^{D} w_j S_j - K, 0 \big)
$$
for a basket call with weights \( w_j \). At the domain boundaries, suitable Dirichlet or asymptotic conditions are imposed:
\[
V(\mathbf{0}, t) = 0, \qquad
V(\mathbf{S}, t) \to \sum_{j} w_j S_j e^{-q_j (T-t)} - K e^{-r (T-t)}
\quad \text{as } S_j \to \infty.
\]
Although, due to speed of information propagation, or the shape of the heat-kernel (Gaussian) the boundaries are irrelevant when they are far away from the domain of interest. 

First we set:
\begin{align}
    x_i &= \log{S_i} \cr
    \tau &= T-t \cr 
    U (\tau,x)&= V(t,S).
\end{align}
With this, we obtain again the same form:
\[
- U_\tau
+ \vec{\gamma}^{\top} \nabla_x U
+ \tfrac{1}{2}\,\nabla_x^{\top} \Sigma \,\nabla_x U
- r U = 0,\qquad x = \log S,\quad \tau = T-t,
\]
with $\Sigma_{ij} = \rho_{jk}\sigma_j \sigma_k$ being semi-positive definite and $\gamma_i = r - q_i$. 

Now, we choose $A$ such that $A\Sigma A^{\top} = I$ (e.g.\ $A = \Lambda^{-1/2}Q^{\top}$ for $\Sigma = Q\Lambda Q^{\top}$), set
\[
z = A x,\qquad \vec{a} = A\vec{\gamma},
\]
and define
\[
U(\tau,x)
= \exp\!\Big(
   -\vec{a}^{\top} z
   - \big(r + \tfrac12\|\vec{a}\|^2\big)\tau
  \Big)\,u(\tau,z).
\]
Then $u$ satisfies the isotropic heat equation
\[
u_\tau = \frac{1}{2}\Delta_z u,
\]
and the payoff becomes
$$
\phi(z) = \exp\!\Big(
   \vec{a}^{\top} z
  \Big) \max\!\left( \sum_{j=1}^{D} w_j e^{(A^{-1} z)_j} - K, 0 \right).
$$
At each Bermudan exercise date \( \tau_i = T - t_i  \), the early exercise condition couples the PDE segments:
\begin{align}
u(z, \tau_i) = 
\max\!\left(
u_{\text{cont}}(z, \tau_i^+),\, \phi(z)
\right),
\label{eq:bermudan_exercise}
\end{align}
where \( u_{\text{cont}}(z, \tau_i^+) \) is the continuation value obtained by backward propagation (now forward propagation because of the change of variables $\tau = T - t$) of Eq.~\eqref{eq:multi_bs_pde} from \( \tau_{i+1} \) to \( \tau_i \). Thus, we solve the differential equation the following way:
\begin{enumerate}
    \item Prepare the quantum memory in the terminal (initial in $\tau$) state $u(z,0)$ through Grover-Rudolph (GR)
    \item Propagate from $\tau_N=T-t_N=0$ to $\tau_{N-1}$
    \item Wrap steps 1 and 2 in a Grover operator to extract the solution $u_{\rm cont}(z,\tau_{N-1}^+)$ within precision $\varepsilon$ using the interpolation methods of ~\cite{GumaroS2026}.
    \item Prepare our approximation of state $u(z, \tau_{N-1}) = \max\!\left(u_{\text{cont}}(z, \tau_{N-1}^+),\,\phi(z)\right)$ through GR. 
    \item Propagate from $\tau_{N-1}$ to $\tau_{N-2}$.
    \item Repeat until reaching $\tau_0=0$.
\end{enumerate}

The expected error at the end is
\begin{align}
    \mathcal{O}(N \varepsilon)
\end{align}
for $\varepsilon N \ll 1$, otherwise one gets an exponential growth of error. The cost of each of the steps is
\begin{align}
    \mathcal{O} \left( 2^{D/2} M^{4D} \frac{1}{\varepsilon}\right),
\end{align}
where the expected smoothness parameter (See \Cref{sec:gaussian_derivative_bounds} ) is $\Lambda = \mathcal{O} \left( \frac{\sqrt{\log{1/\varepsilon}}}{\sqrt{2\kappa \Delta t}}  \right)$, where $\Delta t = \min_j (t_{j}-t_{j-1})$ for $1\leq j\leq N$, and provided the overall polynomial degree, $(M-1)^D$, is greater than $\frac{D}{2\log{2}} \log \left(\frac{1}{4\pi \kappa \Delta t}\right)$.
For a final target error of $\epsilon$, we must choose
\begin{align}
    \varepsilon = {\Theta}\left(\epsilon/N\right),
\end{align}
thus, the final cost scales like
\begin{align}
    \mathcal{O} \left(N^2 2^{D/2} M^{4D} \frac{1}{\epsilon} \right).
\end{align}
With this, we demonstrate a quantum advantage with respect to an asymptotic scaling in $\epsilon$ compared to classical methods like classical Monte Carlo, binomial lattices, which scale like $\mathcal{O}(\epsilon^{-2})$. Even more so for classical approaches to solve PDEs when the problem is multi-asset, where time-complexity scales like $\mathcal{O}(1/\epsilon^{1/2+d/2})$.

Compared to the work in \cite{Miyamoto2022}, which uses quantum Monte Carlo,  our method works exponentially better with respect to the number of exercise dates $N$. The computational complexity does however scale exponentially with respect to the number of assets for both approaches.

\begin{table}[h!]
\centering
\renewcommand{\arraystretch}{1.25}
\begin{tabular}{c|c|c|c}
\hline
$K$ 
& $\displaystyle \frac{e^{0.1K}}{K^{5/2}}$ 
& $\displaystyle \frac{e^{0.5K}}{K^{5/2}}$ 
& $\displaystyle \frac{e^{K}}{K^{5/2}}$ \\
\hline
10  
& $8.60\times 10^{-3}$ 
& $4.68\times 10^{-1}$ 
& $6.97\times 10^{1}$ \\

20  
& $4.13\times 10^{-3}$ 
& $3.00\times 10^{1}$ 
& $2.71\times 10^{5}$ \\

50  
& $8.39\times 10^{-3}$ 
& $3.17\times 10^{6}$ 
& $2.93\times 10^{17}$ \\

100 
& $2.20\times 10^{-1}$ 
& $5.18\times 10^{16}$ 
& $2.69\times 10^{38}$ \\

200 
& $8.58\times 10^{2}$ 
& $1.18\times 10^{39}$ 
& $1.28\times 10^{81}$ \\

500 
& $9.27\times 10^{14}$ 
& $1.29\times 10^{96}$ 
& $2.51\times 10^{210}$ \\
\hline
\end{tabular}
\caption{Ratios of exponential growth from \cite{Miyamoto2022} to polynomial growth in cost from this work: $\frac{e^{\alpha K}}{K^{5/2}}$ for $\alpha=0.1,0.5,1$. We could not ascertain exactly the coefficients of the exponential growth in cost from ~\cite{Miyamoto2022} and we only obtain $\mathcal{O}$-notation bounds for the polynomial growth in this work, so we have chosen to tabulate a few values of $\alpha$ to demonstrate how quickly the exponential growth over-takes the polynomial growth obtained here. The $N^{5/2}$-type scaling is obtained after we limit ourselves to a 1D problem.}
\end{table}

\begin{figure}
    \centering
    \includegraphics[width=0.8\linewidth]{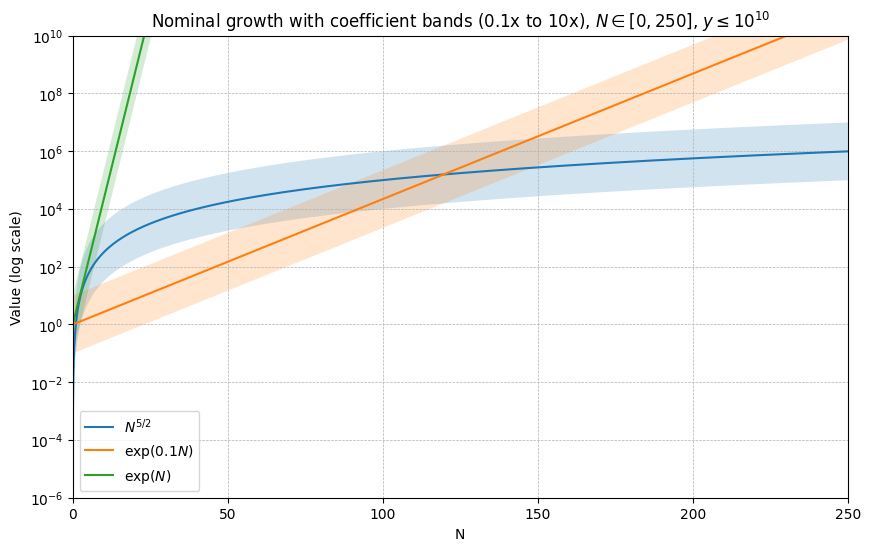}
    \caption{Cost growth comparison with respect to number of exercise dates, $N$, between the method developed here and the one in \cite{Miyamoto2022}. We have varied the coefficients enveloping the functional form of the cost-bounds to illustrate the sensitivity to different likely scenarios as we only have access to $\mathcal{O}$-notation bounds.}
    \label{fig:placeholder}
\end{figure}

\begin{table}[h]
\centering
\begin{tabular}{c c c c}
\hline
$\epsilon$ & $D=1$ & $D=3$ & $D=7$ \\
\hline
$10^{-1}$ & $30$--$40$ & $70$--$94$ & $150$--$200$ \\
$10^{-2}$ & $60$--$70$ & $140$--$164$ & $300$--$350$ \\
$10^{-3}$ & $90$--$100$ & $210$--$234$ & $450$--$500$ \\
$10^{-4}$ & $120$--$130$ & $280$--$304$ & $600$--$650$ \\
\hline
\end{tabular}
\caption{Our estimated qubit requirements for working/storing the solution for Bermudan option pricing plus the ancillary qubits for Fourier LCU. The range indicates the sensitivity of values by varying the proportionality constant of error convergence with respect to $N$. Values of $D\log_2(C/\epsilon^6)+\log_2(\sqrt{C}/\epsilon^3)$ for $C\in[1,100]$, rounded up to the nearest integer.}
\end{table}

\section{Conclusion and Outlook}

We have come up with new methods for solution preparation for the heat equation, exponentially improving existing methods. The methods developed here can be also advantageously used for other PDEs where the discretized operator is fast-forwardable. We have also demonstrated that this is an exemplary use case for the solution extraction methods developed in ~\cite{GumaroS2026}. We hope to also open the possibility for quantum advantage for other types of problems, here we have just listed a few examples. The solution extraction remains the bottleneck for problems with very high number of dimensions since the cost scaling remains exponential with $D$ (number of dimensions). This type of scaling is also expected for quantum Monte Carlo methods if one wants to get the whole solution surface since Monte Carlo methods only get you a point of the surface and interpolation can be used but at an exponential cost with $D$ as well. Most often than not Monte Carlo methods avoid this by not getting the whole surface and evaluating point-by-point. We also demonstrate an exponential improvement in cost with respect to the number of exercise dates compared to other quantum methods for pricing Bermudan options (\Cref{subsec:Bermudan}).
We hope that in the future we can address this ameliorated dimensionality curse for the PDE approach. Moreover, we hope to extend these methods to other kind of options like some other volatility method like the Heston model.

\appendix

\section{Derivative bounds on Approximate Gaussian Distribution}\label{sec:gaussian_derivative_bounds}

The Gaussian function is the fundamental solution to the Heat equation. That is, it describes the response of the equation to the Dirac delta impulse, $\delta(\vec{x})$. Thus, all solutions inherit the smoothness of this solution (after some time $t$).

\paragraph{Heat kernel.}
For $\kappa>0$ and $t>0$, the $D$-dimensional heat kernel is
\[
G_t(x)
=\frac{1}{(4\pi \kappa t)^{D/2}}
\exp\!\left(-\frac{\abs{x}^2}{4\kappa t}\right),\qquad x\in\R^D,
\]
so that $\int_{\R^D} G_t(x)\,dx=1$ and $G_t*G_s=G_{t+s}$.

\paragraph{Cauchy problem (homogeneous).}
Let $u_0\in L^1(\R^D)$ (or a finite measure / tempered distribution).
The unique solution of
\[
\partial_t u = \kappa \Delta u\quad\text{in }\R^D\times(0,\infty),\qquad
u(\cdot,0)=u_0
\]
is given for $t>0$ by the Gaussian convolution
\[
u(x,t) \;=\; (G_t*u_0)(x)
\;=\; \int_{\R^D} G_t(x-y)\,u_0(y)\,dy.
\]

Assume $u\in L^1(\mathbb{R})$ and, for all $n\ge 0$ and $x\in\mathbb{R}$,
\[
\bigl|f^{(n)}(x)\bigr|\le \Lambda^{\,n+1}.
\]
Then, for every $n\ge 0$ and $x\in\mathbb{R}$,
\[
\;
(u*f)^{(n)}(x)
= (u * f^{(n)})(x)
\le \|u\|_{L^1}\,\|f^{(n)}\|_{L^\infty}
\le \|u\|_{L^1}\,\Lambda^{\,n+1}. \;
\]

\emph{Reason.} Differentiate under the integral and apply Young’s inequality $L^1*\!L^\infty\to L^\infty$:
\begin{align*}
|(u*f)^{(n)}(x)|
&= \left| \frac{d^n}{dx^n}\int_{\mathbb{R}} u(y)\,f(x-y)\,dy \right|
 = \left| \int_{\mathbb{R}} u(y)\,f^{(n)}(x-y)\,dy \right| \\
&\le \int_{\mathbb{R}} |u(y)|\,|f^{(n)}(x-y)|\,dy
 \le \|u\|_{L^1}\,\|f^{(n)}\|_{L^\infty}
 \le \|u\|_{L^1}\,\Lambda^{\,n+1}.
\end{align*}

However, for the Gaussian function $g_\sigma(x) = e^{-x^{2}/(2\sigma^{2})}$, a simple derivative bound estimate that follows from Cauchy’s estimate also gives a $\left(\tfrac{n}{e}\right)^{n/2}$ growth factor which does not fit the required form in ~\Cref{lem:SE_multidim}

We can quickly show this by first using the bound from Cauchy's estimate:
\begin{align}
    \sup_x |G_\sigma^{(n)}(x)|
    &\le \frac{n!}{r^n} e^{r^{2}/(2\sigma^{2})}
\end{align} 
Optimizing at $r^{2}=n\sigma^{2}$ yields
\begin{align}
    \sup_x |G_\sigma^{(n)}(x)|
    \le \frac{n!\,e^{n/2}}{\sigma^{n} n^{n/2}}
    \;\sim\;
    \frac{\sqrt{2\pi n}}{\sigma^{n}}\left(\frac{n}{e}\right)^{n/2}.
\end{align}
Thus, the $n$-th derivative grows roughly as
\[
|g_\sigma^{(n)}|\;\lesssim\;\sigma^{-n}\left(\tfrac{n}{e}\right)^{n/2},
\]
up to polynomial factors in $n$.

However, we can bound the derivatives of a proxy distribution which approximates the Gaussian. Thus, we can think of the Solution Extraction (SE) method to be extracting this approximate Gaussian distribution (or the initial condition convoluted with the approximate Gaussian) with the number of samples required converging quickly at the price of increasing the error by $\| u\|_{L^1} \epsilon$, where $\epsilon$ is the error between the Gaussian and this proxy distribution.  One can use Young's inequality to see this:
\begin{align}
    \max_{x} \left|(u*\hat{G}_{t} - u*\hat{G}_{t;\epsilon})(x)\right| \leq \| u \|_{L^{1}} \| \hat{G}_{t} - \hat{G}_{t;\epsilon}\|_{\infty} = \epsilon\leq \| u \|_{L^{1}}.
\end{align}

To find this proxy distribution whose derivatives we can bound through a bound of the form $\Lambda^{n+1}$, we borrow a slightly modified lemma and proof from ~\cite{WangFRJ2025} that provides an efficient trigonometric expansion for a centered Gaussian distribution:

\begin{lemma}[Trigonometric Approximation of the Gaussian Distribution]

For every $\sigma, \epsilon \in (0, 1)$, there exist an efficiently-computable trigonometric polynomial $G_{\sigma;\epsilon}(x) = \sum_{j=0}^N a_j \cos(2\pi j x/T)$ of degree $N=\mathcal{O}(\sigma^{-1} T \sqrt{\mylog{1/\epsilon}})$, where $T=\max(2\pi, \Theta(\sigma \sqrt{\mylog{1/\epsilon}}))$, such that 
\begin{itemize}
    \item $\sum^{N}_{j=0} |a_j| \le 1$, for all $x \in \R$;
    \item $|G_{\sigma;\epsilon}(x) - e^{-x^2/(2\sigma^2)}| \le \epsilon$, for all $x \in [-\pi, \pi]$.
\end{itemize}
\label{lem:trigonometric_approx_gaussian_func}
\end{lemma}

\begin{proof}[Proof of Lemma \ref{lem:trigonometric_approx_gaussian_func}]
    We claim that for every $\sigma, \epsilon \in (0, 1)$, there exist $T=\max(2\pi, \Theta(\sigma \sqrt{\mylog{1/\epsilon}}))$, $N=\mathcal{O}(\sigma^{-1} T \sqrt{\mylog{1/\epsilon}})$ and $a_0, a_1, \dots, a_N \in \R^+$ such that $\sum_{j=0}^N a_j \le 1$ and  
   \begin{align}
   \abs{\sum_{j=0}^N a_j \cos(2\pi j x/T) - e^{-x^2/(2 \sigma^2)}} \le \epsilon,~&~\forall x \in [-\pi, \pi].       
   \end{align}

The proof of the above claim is inspired by the Nyquist–Shannon sampling theorem. Specifically, let $f(x) \defeq \sum_{k=-\infty}^\infty \delta(x-k T)$ be the Dirac comb with period $T$. It has Fourier transform $\hat{f}(\xi) = \frac{1}{T}\sum_{n=-\infty}^\infty \delta(\xi - n/T)$. Meanwhile, let $g(x) \defeq e^{-x^2/(2\sigma^2)}$ be Gaussian with mean $0$ and variance $\sigma^2$. It has Fourier transform $\hat{g}(\xi)=\sqrt{2\pi}\sigma e^{-2\pi^2 \sigma^2 \xi^2}$. Convolving $f$ and $g$ yields
\begin{align}
(f * g)(x) = \sum_{k=-\infty}^\infty g(x-kT).
\label{eq:conv1}
\end{align}
On the other hand, the Fourier transform of $f*g$ is 
\begin{align}
\widehat{f*g}(\xi) = \hat{f}(\xi) \hat{g}(\xi)=  \frac{1}{T}\sum_{n=-\infty}^\infty  \hat{g}(n/T) \delta(\xi - n/T).
\end{align}
Applying inverse Fourier transform to both sides of this equation leads to
\begin{align}
(f * g)(x) = \frac{1}{T}\sum_{n=-\infty}^\infty  \hat{g}(n/T) e^{\i 2\pi x n/T}. 
\label{eq:conv2}
\end{align}
Comparison of Eqs.~\eqref{eq:conv1} and \eqref{eq:conv2} indicates that
\begin{align}
    \sum_{k=-\infty}^\infty g(x-kT) =  \frac{1}{T}\sum_{n=-\infty}^\infty  \hat{g}(n/T) e^{\i 2\pi x n/T}. 
    \label{eq:conv3}
\end{align}
Now we claim that for some $T=\max(2\pi, \Theta(\sigma \sqrt{\mylog{1/\epsilon}}))$ and $N={\mathcal O}(\sigma^{-1} T \sqrt{\mylog{1/\epsilon}})$, 
\begin{align}
     \sum_{k=-\infty}^{-1} g(x-kT) + \sum_{k=1}^\infty g(x-kT)  \le \frac{\epsilon}{6}, ~&~\forall x \in [-\pi, \pi],
\label{eq:trunc_err1}
\end{align}
and
\begin{align}
    \abs{ \frac{1}{T}\sum_{n=-\infty}^{-N-1}  \hat{g}(n/T) e^{\i 2\pi x n/T}
+    \frac{1}{T}\sum_{n=N+1}^\infty  \hat{g}(n/T) e^{\i 2\pi x n/T}}
\le  \frac{\epsilon}{6},~&~\forall x \in \R.
\label{eq:trunc_err2}
\end{align}
If these claims are true, then combining them and Eq. \eqref{eq:conv3} yields
\begin{align}
\abs{ g(x) - \frac{1}{T}\sum_{n=-N}^N  \hat{g}(n/T) e^{\i 2\pi x n/T}}
=\abs{ g(x) - \hat{g}(0) - \frac{2}{T}\sum_{n=1}^N  \hat{g}(n/T) \cos(2\pi x n/T)} \le \frac{\epsilon}{3}, ~&~\forall x \in [-\pi, \pi].
\end{align}
This also implies that 
\begin{align}
\hat{g}(0) + \frac{2}{T}\sum_{n=1}^N  \hat{g}(n/T) \le 1+\frac{\epsilon}{3}.    
\end{align}
Now let $a_0=(1-\epsilon/3)\hat{g}(0)$ and $a_j=2T^{-1}(1-\epsilon/3)\hat{g}(j/T)$ for $j=1,2,\dots,N$. 
Then we have $a_0, a_1, \dots, a_N >0$ and
\begin{align}
    \sum_{j=0}^N a_j = (1-\epsilon/3) \lrb{\hat{g}(0) + \frac{2}{T}\sum_{n=1}^N  \hat{g}(n/T)} \le (1-\epsilon/3)(1+\epsilon/3) < 1,
\end{align}
and for any $x \in [-\pi, \pi]$,
\begin{align}
       \abs{\sum_{j=0}^N a_j \cos(2\pi j x/T) - g(x)}
       &\le
       \abs{\sum_{j=0}^N a_j \cos(2\pi j x/T) - \hat{g}(0) - \frac{2}{T}\sum_{n=1}^N  \hat{g}(n/T) \cos(2\pi x n/T)} \nonumber\\
       &\quad + \abs{\hat{g}(0) + \frac{2}{T}\sum_{n=1}^N  \hat{g}(n/T) \cos(2\pi x n/T) - g(x)}\\
        &\le \frac{\epsilon}{3} \cdot \abs{ \hat{g}(0) + \frac{2}{T}\sum_{n=1}^N  \hat{g}(n/T) \cos(2\pi x n/T)} + \frac{\epsilon}{3}\\
       &\le \frac{\epsilon}{3} \lrb{\hat{g}(0) + \frac{2}{T}\sum_{n=1}^N  \hat{g}(n/T) } + \frac{\epsilon}{3} \\
       &\le \frac{\epsilon}{3} \lrb{1+\frac{\epsilon}{3}} + \frac{\epsilon}{3}\\
       &\le \epsilon,
\end{align}
as desired.

It remains to prove Eqs.~\eqref{eq:trunc_err1} and \eqref{eq:trunc_err2}. To prove the former, we impose the constraint $T \ge 2\pi$. This means that $|x|\le T/2$ for all $x \in [-\pi, \pi]$. Then it follows that
\begin{align}
    \sum_{k=-\infty}^{-1} g(x-kT) + \sum_{k=1}^\infty g(x-kT)  
   & \le 2 \sum_{l=0}^{\infty} g((l+1/2)T)\\
   & = 2 \sum_{l=0}^{\infty} \myexp{-(l+1/2)^2T^2/(2 \sigma^2)} \\
   & = 2\myexp{-T^2/(8 \sigma^2)} \sum_{l=0}^{\infty} \myexp{-(l^2+l)T^2/(2 \sigma^2)} \\
   & \le 2\myexp{-T^2/(8 \sigma^2)} \sum_{l=0}^{\infty} \myexp{-l T^2/\sigma^2} \\
   & = \frac{2\myexp{-T^2/(8 \sigma^2)} }{1 - \myexp{-T^2/\sigma^2} }.
   \label{eq:trunc_err1_bound}
\end{align}
By picking some $T=\Theta(\sigma \sqrt{\mylog{1/\epsilon}})$, we can ensure that the RHS of Eq.~\eqref{eq:trunc_err1_bound} is at most ${\epsilon}/{6}$. Overall, Eqs.~\eqref{eq:trunc_err1} holds as long as we pick some $T=\max(2\pi, \Theta(\sigma \sqrt{\mylog{1/\epsilon}}))$. 

To prove Eq. \eqref{eq:trunc_err2}, we note that
\begin{align}
    \abs{ \frac{1}{T}\sum_{n=-\infty}^{-N-1}  \hat{g}(n/T) e^{\i 2\pi x n/T}
+    \frac{1}{T}\sum_{n=N+1}^\infty  \hat{g}(n/T) e^{\i 2\pi x n/T}}
&\le      \frac{1}{T}\sum_{n=-\infty}^{-N-1}  \hat{g}(n/T)
+    \frac{1}{T}\sum_{n=N+1}^\infty  \hat{g}(n/T) \\
&= \frac{2}{T}\sum_{n=N+1}^\infty  \hat{g}(n/T) \\
&= \frac{2\sqrt{2\pi} \sigma}{T} \sum_{n=N+1}^\infty \myexp{-2 \pi^2 \sigma^2 n^2/T^2} \\
&\le \frac{2\sqrt{2\pi} \sigma}{T} \int_{N}^\infty \myexp{-2 \pi^2 \sigma^2 y^2/T^2} dy \\
&= \frac{2}{\sqrt{\pi}} \int_{\sqrt{2}\pi \sigma N/T}^\infty \myexp{-z^2} dz \\
&= \erfc(\sqrt{2}\pi \sigma N/T) \\
&\le \myexp{-2 \pi^2 \sigma^2 N^2 / T^2}. 
\end{align}
By picking some $N=\mathcal{O}(\sigma^{-1} T \sqrt{\mylog{1/\epsilon}})$, we can guarantee that
$\myexp{-2 \pi^2 \sigma^2 N^2 / T^2} \le \epsilon/6$ and hence Eq. \eqref{eq:trunc_err2} holds for all $x \in \R$.

\end{proof}

We can now bound the $n$-th derivative of this proxy distribution:

\begin{lemma}[Derivative bounds]

Provided there exists an efficiently-computable trigonometric polynomial $G_{\sigma;\epsilon}(x) = \sum_{j=0}^N a_j \cos(2\pi j x/T)$ of degree $N=\mathcal{O}(\sigma^{-1} T \sqrt{\mylog{1/\epsilon}})$, where $T=\max(2\pi, \Theta(\sigma \sqrt{\mylog{1/\epsilon}}))$, and such that 
\begin{itemize}
    \item $\sum^{N}_{j=0} |a_j| \le 1$, for all $x \in \R$;
    \item $|G_{\sigma;\epsilon}(x) - e^{-x^2/(2\sigma^2)}| \le \epsilon$, for all $x \in [-\pi, \pi]$
\end{itemize}

For the function $G_{\sigma; \epsilon}(x)=\sum^{N}_{j=0} a_j \mycos{2\pi j x/T} $ that is $\epsilon$ close to $g(x)=e^{-x^2/(2\sigma^2)}$ has the following derivative bounds
\begin{align}
    \left|\frac{{\rm d}^n G_{\sigma;\epsilon}}{{\rm d} x^n}\right| \leq \left(\Lambda\right)^{n+1}
\end{align}
for
\begin{align}
    \Lambda = \frac{2\pi N}{T}= \mathcal{O} \left(\frac{\sqrt{log(1/\epsilon)}}{\sigma}\right),
\end{align}
provided $2\pi N /T \geq 1$.
\end{lemma}

\begin{proof}
    If we have the function form the previous theorem
    \begin{align}
        f(x) =\sum^{N}_{j=0} a_j \mycos{2\pi j x /T}
    \end{align}
    
    \begin{align}
        \frac{{\rm d} G_{\sigma; \epsilon}}{ {\rm d} x} = \sum_j a_j \frac{2\pi j} {T} (-1) \mysin{2\pi j x/T} 
    \end{align}
    \begin{align}
        \frac{{\rm d}^2 G_{\sigma; \epsilon}}{ {\rm d} x^2} = \sum_j a_j \left(\frac{2\pi j} {T}\right)^2 (-1) \mycos{2\pi j x/T} 
    \end{align}
    By induction:
   \begin{align}
        \frac{{\rm d}^n G_{\sigma; \epsilon}}{ {\rm d} x^n} = \sum_j a_j \left(\frac{2\pi j} {T}\right)^n (-1)^{(n+1)/2} \mysin{2\pi j x/T} \quad \text{for odd } n
    \end{align}
    and
     \begin{align}
        \frac{{\rm d}^n G_{\sigma; \epsilon}}{ {\rm d} x^n} = \sum_j a_j \left(\frac{2\pi j} {T}\right)^n (-1)^{n/2} \mycos{2\pi j x/T} \quad \text{for even } n.
    \end{align}
    Thus,
    \begin{align}
        \left|  \frac{{\rm d}^n G_{\sigma; \epsilon}}{ {\rm d} x^n}\right| \leq \left(\frac{2\pi N}{T}\right)^n,
    \end{align}
    and for $2\pi N/T \geq 1$, we have
    \begin{align}
        \left|  \frac{{\rm d}^n G_{\sigma; \epsilon}}{ {\rm d} x^n}\right| \leq \Lambda^{n+1} ,
    \end{align}
    where
    \begin{align}
        \Lambda = \frac{2\pi N}{T} = \mathcal{O} \left(\frac{\sqrt{\log(1/\epsilon)}}{\sigma}\right)
    \end{align}
    \end{proof}

\subsection{Generalization to D-dimensions}

\begin{lemma}[Trigonometric Approximation of  the $D$-dimensional Gaussian Distribution]

For every $\sigma \in (0, 1)$ and sufficiently small $\epsilon > 0$, there exists an efficiently-computable $D$-variable trigonometric polynomial 
\begin{align*}
G_{\sigma;\epsilon;D}(\vec{x}) &= \prod^{D}_{i=1} G_{\sigma;\varepsilon}(x_i)\cr
&=\prod^{D}_{i=1}\left(\sum_{j=0}^N a_j \cos(2\pi j x_i/T)\right)
\end{align*}
of degree $N=\mathcal{O}(\sigma^{-1} T \sqrt{\mylog{1/\varepsilon}})$ on each of the variables, where $T=\max(2\pi, \Theta(\sigma \sqrt{\mylog{1/\varepsilon}}))$, for some
\begin{align*}
    \varepsilon = \mathcal{O}(\epsilon/D)
\end{align*}
such that 
\begin{itemize}
    \item $\sum^{N}_{j=0} |a_j| \le 1$, for all $x \in \R$;
    \item $|G_{\sigma;\epsilon;D}(\vec{x}) - e^{-|\vec{x}|^2/(2\sigma^2)}| \le \epsilon$, for all $x_i \in [-\pi, \pi]$.
\end{itemize}
\label{lem:trigonometric_approx_gaussian_func}
\end{lemma}

Finally, the derivative bounds for this $D$-dimensional proxy distribution is:

\begin{lemma}[Multivariate Derivative Bound]

\begin{align}
    \left\|\partial^{\alpha} G_{\sigma;\epsilon;D}\right\|_{\infty} \leq \Lambda^{|\alpha|+1},
\end{align}
where
    \begin{align}
        \Lambda = \frac{2\pi N}{T} = \mathcal{O} \left(\frac{\sqrt{\log(D/\epsilon)}}{\sigma}\right),
    \end{align}
provided that $2\pi N / T \geq 1$.
\end{lemma}

We now write a similar statement in terms of an almost-normalized version of the proxy distribution $\hat{G}_{t;\epsilon;D}$, where
\begin{align}
    \hat{G}_{t;\epsilon;D} \left(\vec{x}\right) = \frac{1}{(4\pi \kappa t)^{D/2}} G_{\sqrt{2\kappa t};(4\pi \kappa t)^{D/2}\epsilon;D} \left(\vec{x}\right).
\end{align}

That is, the corresponding standard deviation in terms of $\kappa$ and $t$ is $\sigma=\sqrt{2\kappa t}$. With this, the multivariate derivative bound becomes:

\begin{lemma}[Multivariate Derivative Bound]

\begin{align}
    \left\|\partial^{\alpha} \hat{G}_{t;\epsilon;D}\right\|_{\infty} \leq {\Lambda}^{|\alpha|+1},
\end{align}
where
    \begin{align}
        \Lambda = \frac{4\pi N}{T} =  \mathcal{O} \left(\frac{\sqrt{\mylog{\frac{D}{\epsilon}}}}{\sqrt{2 \kappa t}} \right),
\end{align}
provided that $2\pi N / T \geq 1$ and $|\alpha|\geq  {\frac{D}{2\log 2}}\log\left( \frac{1}{4\pi \kappa t}\right)-1$.
\end{lemma}
\begin{proof}
We have that
\begin{align}
    \left\|\partial^{\alpha} \hat{G}_{t;\epsilon;D}\right\|_{\infty} \leq \frac{1}{(4\pi \kappa t)^{D/2}}{\Lambda}^{|\alpha|+1},
\end{align}
where $\sigma = \sqrt{2 \kappa t}$, and
    \begin{align}
        \Lambda = \frac{2\pi N}{T} = \mathcal{O} \left(\frac{\sqrt{\log(D/\epsilon)}}{\sqrt{2\kappa t}}\right),
    \end{align}
thus, we can bound the derivatives with:
\begin{align}
\left\|\partial^{\alpha} \hat{G}_{\sigma;\epsilon;D}\right\|_{\infty} \leq \left(\frac{1}{(4\pi \kappa t)^{\frac{D}{2(|\alpha|+1)}}}{\Lambda}\right)^{|\alpha|+1}.
\end{align}

We now impose the constraint:
\begin{align}
    \frac{1}{(4\pi \kappa t)^{\frac{D}{2(|\alpha|+1)}}} \leq 2.
\end{align}
Thus, the following bound:
\begin{align}
\left\|\partial^{\alpha} \hat{G}_{\sigma;\epsilon;D}\right\|_{\infty} \leq \left(2{\Lambda}\right)^{|\alpha|+1},
\end{align}
holds when
\begin{align}
 (|\alpha|+1) \geq  {\frac{D}{2\log 2}}\log\left( \frac{1}{(4\pi \kappa t)}\right).
\end{align}

\end{proof}

\section{Heat-type equations in pricing derivatives}
This section will delve into transforming PDEs associated with pricing derivatives to a heat-type equation. We use the term ``heat-type'' equation to describe a PDE of the form described in equation~\ref{eq:fkacPDE} where $b$ is a time-dependent vector and $A$ is a time-dependent diagonal matrix. In other words, the coefficients of the second-order cross-term derivatives are $0$.

To begin, a derivative is a financial contract between two parties whose value is derived from an underlying asset. The value is defined by a payoff function $g(X_T)$ where $T$ is the time to maturity and $X_t$ is the stochastic process governing the price of the asset. Derivatives exist in different asset classes, some of which we will encounter shortly, and they all share a common valuation principle: the present value of a derivative equals the expected discounted value of its payoff under a probability measure in which the appropriately scaled underlying stochastic process is a martingale. This allows us to employ Feynman-Kac formula to generate a PDE solving for the value of the derivative. Having given a gentle introduction to derivatives, we now turn to individual asset classes where PDE pricing of certain types of derivatives are amenable to our quantum techniques. We will rigorously demonstrate the transformation of the PDEs into heat-type equation. We note that such transformations are already demonstrated in literature; however, for the benefit of the readers unfamiliar with these methods, we will show the transformation steps for derivatives where our quantum techniques can be applied.

\subsection{Equity derivatives}
For an equity asset, the price $S_t$ is assumed to follow a stochastic differential equation (SDE) of the form
\[
dS_t = \mu(t)\,S_t\,dt + \sigma(t)\,S_t\,dW^{\mathbb P}_t,
\]
where $W^{\mathbb P}_t$ is a Wiener process under the so-called real-world measure $\mathbb P$, $\mu(t,S_t)$ is the instantaneous drift or expected rate of return, and $\sigma(t,S_t)$ is the volatility. Both coefficients may vary deterministically with time but are assumed to be independenf of the current price of the underlying asset.

If we consider a derivative with maturity $T$ and payoff function $g(S_T)$, its price at an earlier time $t$ is the discounted expected value of that payoff under a suitable measure. To define this measure, we introduce the \emph{money-market account}
\[
B_t = \exp\left(\int_0^t r(u)\,du\right),
\]
which accumulates continuously at the instantaneous risk-free rate $r(t)$. This asset plays a special role: in an arbitrage-free market, the discounted price process $S_t / B_t$ must be a martingale under some probability measure $\mathbb Q$. This measure is called the \emph{risk-neutral measure}.

To construct $\mathbb Q$, start from the SDE under the real-world measure $\mathbb P$, and use Girsanov's theorem to define a new Wiener process under $\mathbb Q$:
\[
dW_t^{\mathbb Q} = dW_t^{\mathbb P} + \theta_t\,dt,
\]
where the process $\theta_t = \tfrac{\mu(t)-r(t)}{\sigma(t)}$. Substituting this into the SDE gives the risk-neutral dynamics
\begin{equation}\label{eq:riskneutralBS-SDE}
dS_t = r(t)\,S_t\,dt + \sigma(t)\,S_t\,dW_t^{\mathbb Q}.
\end{equation}
Under this measure, the discounted process $\tilde S_t = S_t/B_t$ satisfies $d\tilde S_t = \tilde S_t\,\sigma(t)\,dW_t^{\mathbb Q}$ and therefore has zero drift. The process $\tilde S_t$ is thus a martingale under $\mathbb Q$.

Let $V(t,S)$ denote the value at time $t$ of a derivative whose payoff at maturity is $g(S_T)$. Because the discounted derivative value $\exp\left(-\int_0^t r(u)\,du\right)V(t,S_t)$ must itself be a martingale, from Theorem~\ref{thm:fkac}, we have the following PDE
\begin{equation}\label{eq:BSPDE}
\frac{\partial V}{\partial t}
+ r(t)\,S\,\frac{\partial V}{\partial S}
+ \tfrac{1}{2}\,\sigma^2(t)\,S^2\,\frac{\partial^2 V}{\partial S^2}
- r(t)\,V = 0
\quad \text{such that} \quad
V(T,S) = g(S).
\end{equation}

Now let us turn our attention to a few type of options in the equity market. We assume that the volatility of the underlying asset is constant throughout the life of the option. We begin with European options, which are arguably the simplest and most fundamental options. 

A \textbf{European call option} on an underlying asset with price $S_t$ and strike $K$ grants its holder the right, but not the obligation, to purchase the asset for $K$ at the maturity date $T$. Its payoff is therefore
\[
g_{\mathrm{call}}(S_T) = (S_T-K)^+ = \max(S_T - K,\,0).
\]
Similarly, a \textbf{European put option} provides the right to sell the asset for $K$ at maturity, yielding payoff
\[
g_{\mathrm{put}}(S_T) = (K - S_T)^+ = \max(K - S_T,\,0).
\]
In both cases, the option value at an earlier time $t$ is the risk-neutral expectation of the discounted payoff,
\[
V_{\mathrm{call/put}}(t,S_t) 
= \mathbb E^{\mathbb Q}\!\left[\exp\left(-\int_t^T r(u)\,du\right)\,g_{\mathrm{call/put}}(S_T)\,\big|\,S_t\right].
\]
This expectation satisfies the fundamental pricing PDE introduced earlier. While we consider payoff function $g(.)$ only for call and put options, one can have other functions as well as long as their growth is linearly bounded.

\subsubsection{Bermudan options}

Now we move to the Bermudan option which also admits a PDE that can be transformed to a heat equation, and hence, can be solved using our quantum algorithm. A Bermudan option is an option that has a series of discrete exercise dates
\[
0 = t_0 < t_1 < \dots < t_{n-1} < t_n = T
\]
with a payoff function $g(S)$ which could be $(S-K)^+$ for a call option and $(K-S)^+$ for a put option. Let's denote the set of exercise dates as $\mathcal{T} := \{t_0,t_1, \dots,t_n\}$ and the set of exercise dates starting from $t_k$ as $\mathcal{T}_k := \{t_k, \dots,t_n\}$. The pricing of the Bermudan option is an optimal stopping point problem, where the goal is to find the exercise date among $\mathcal T$ which maximizes the payoff function. The price of the Bermudan option at $t_k$ is
\[
V_k = \sup_{\tau\in\mathcal{T}_k}
\mathbb{E}^{\mathbb{Q}}\!\Big[
 \exp\left({-\int_{t_k}^{\tau} r(u)\,du}\right)\,g(S_\tau)
 \,\Big|\,
 S_{t_k}=S
\Big].
\]
Note that we discount back the price from the optimal exercise date in the future rather than the maturity date. This can be split into two parts: one where the option is exercised at $t_k$ and the other where the option will be exercised at one of the future exercise dates. This implied that the option price at $t_k$ is
\[
V_k(S) = 
\max\left\{
g(S), \sup_{\tau\in\mathcal{T}_{k+1}}
\mathbb{E}^{\mathbb{Q}}\!\Big[
 \exp\left({-\int_{t_k}^{\tau} r(u)\,du}\right)\,g(S_\tau)
 \,\Big|\,
 S_{t_k}=S
\Big]
\right\}
\]
A simple calculation shows that $V_k(S) = 
\max\{g(S), C_k(S)\}$ where
\[
C_k(S) :=
\mathbb{E}^{\mathbb{Q}}
\Big[
 \exp\left(-\int_{t_k}^{t_{k+1}} r(u)\,du\right)\,
 V_{k+1}(S_{t_{k+1}})
 \,\Big|\,
 S_{t_k}=S
\Big].
\]
What it means is that we are pricing all the way to maturity date $T$, but we compress all the information about the price in the value $V_{k+1}(S)$, and then view the Bermudan option strictly within two exercise dates $t_k$ and $t_{k+1}$.

Now, from the time interval $t_{n-1}$ and $t_n = T$, the Bermudan option is exactly an European option, and hence can be solved via the PDE
\[
\frac{\partial V}{\partial t}
+ r(t)\,S\,\frac{\partial V}{\partial S}
+ \tfrac{1}{2}\,\sigma^2(t)\,S^2\,\frac{\partial^2 V_n}{\partial S^2}
- r(t)\,V = 0
\quad \text{such that} \quad
V(T,S) = g(S).
\]
For any two consecutive exercise dates $t_k$ and $t_{k+1}$, the price of the option $V_k$ can be calculated using the same PDE as above, but with a different terminal condition. Since we are proceeding backward from the time of maturity, at any time between $(t_k, t_{k+1})$, we already know the solution of the PDE on the interval $(t_{k+1},t_{k+2})$, which is (via Feynman-Kac formula) given by $C_{k+1}(S)$. Now, we set $V_{k+1} = \max\{g(S), C_{k+1}(S)\}$ and solve the above mentioned PDE with terminal condition $V_{k+1}$ as the solution of the PDE for the time interval $(t_k, t_{k+1})$. We do it until $K = 0$, which would finally yield the price of the Bermudan option.

Given that the PDE of the Bermudan option  is exactly the PDE of an European option between any two consecutive exercise dates, it an be converted to a heat equation with time-dependent spatial term and a different terminal condition.

\subsection{Interest-rate derivatives}

Moving to interest-rate derivatives we will be working with different stochastic processes that aim to model the rate dynamics. One of the models that have gained significant traction in the last decade is the LIBOR Marker Model (LMM). Our foray into this model will start with describing the most fundamental derivative in this asset class: zero coupon bonds (ZCBs). One can find the information on ZCBs in the excellent textbook of Brigo and Mercurio~\cite{BrigoM2006}. Our motivation of describing them is to ensure that the readers can understand the subsequent derivatives, their interplay with LIBOR ratesm and the model itself.

A {\it zero-coupon bond} maturing at time $T$ is a contract that pays exactly one unit of currency at time $T$ and nothing before. 
Its time-$t$ price, $P(t,T)$, is called the \emph{discount factor} from $T$ to $t$:
\[
P(t,T) \in (0,1], \qquad t \le T, 
\qquad P(T,T) = 1.
\]
In a risk-neutral short-rate framework, if $(r_s)_{s \ge 0}$ is the short rate, then
\[
P(t,T) =
\mathbb{E}^{\mathbb{Q}}
\left[
\exp\!\left(-\int_t^T r_s \, ds\right)
\middle| \mathcal{F}_t
\right].
\]
Fix a tenor structure
\[
0 = T_0 < T_1 < \dots < T_N < T_{N+1},
\]
and define the year fractions as
\[
\alpha_i := T_{i+1} - T_i, \qquad i = 0,\dots,N.
\]
Now we can let $P(t,T_k)$ denote the price of a zero-coupon bond (ZCB) at time $t$ and maturing at $T_k$. For each tenor interval $[T_i,T_{i+1}]$, the forward LIBOR rate $L_i(t)$, for $t \le T_i$, is defined by the relationship
\begin{equation}
L_i(t)
:= \frac{1}{\alpha_i}\left(\frac{P(t,T_i)}{P(t,T_{i+1})} - 1\right),
\qquad \text{for all} \qquad i=\{0,\dots,N\}.
\label{eq:def-Li}
\end{equation}
Baked in the equation for $L_i(t)$ is the assumption that this rate is defined at time $t$ for the tenor $[T_i, T_{i+1}]$. One uses the notation $L(t; T_i, T_{i+1})$ for such rates. However, for brevity, we have suppressed the notation to $L_i(t)$ since we have fixed the tenors. From the above equation, we can obtain the ZCB price $P(t,T_{i+1})$ as
\begin{equation}
P(t,T_{i+1})
= \frac{P(t,T_i)}{1 + \alpha_i L_i(t)}.
\label{eq:bond-recursion-one-step}
\end{equation}
By iterating \eqref{eq:bond-recursion-one-step}, one obtains
\begin{equation}
P(t,T_k)
= P(t,T_{N+1}) \prod_{j=k}^{N}\frac{1}{1+\alpha_j L_j(t)},
\qquad \text{for all} \qquad k=\{0,\dots,N\}.
\label{eq:bond-recursion-terminal}
\end{equation}
This is the usual LMM relationship between ZCB prices and the tuple of forward LIBOR rates $L(t) := (L_0(t),\dots,L_N(t))$.

Now we are ready to define the LIBOR Market Model. It models the forward LIBOR rates $\{L_i(t)\}_{i=0}^N$ as lognormal under their respective {\it forward measures}. For each $i$, define the $T_{i+1}$-forward measure $\mathbb{Q}^{i+1}$ associated with the numeraire $P(t,T_{i+1})$. Under $\mathbb{Q}^{i+1}$ it is postulated that $L_i$ follows
\begin{equation}
dL_i(t)
= L_i(t) \,\bm{\sigma}_i(t) \cdot d\bm{W}_t^{i+1},
\qquad 0 \le t \le T_i.
\label{eq:LMM-forward-measure}
\end{equation}
Here, $\bm{W}^{i+1}$ is a $d$-dimensional Brownian motion under $\mathbb{Q}^{i+1}$ and $\bm{\sigma}_i(t)\in\mathbb{R}^d$ is the (row) volatility vector of $L_i$. Moreover, 
\[
\bm{W}_t = \left(W_t^{(1)}, \dots W_t^{(d)}\right)
\]
have independent increments meaning $\ip{dW_t^{(i)}}{dW_t^{(j)}} = \delta_{ij}dt$. For the remainder of this section, we assume that the instantaneous volatilities are {\it piecewise constant} in time. More precisely, let
\[
0 = \tau_0 < \tau_1 < \dots < \tau_M
\]
be a partition of $[0,T_{N+1}]$, and assume
\begin{equation}
\bm{\sigma}_i(t) = \bm{\sigma}_i^{(k)}, \qquad t\in(\tau_k,\tau_{k+1}],
\label{eq:piecewise-const-vol}
\end{equation}
for constants $\bm{\sigma}_i^{(k)}\in\mathbb{R}^d$. This is a very common assumption in literature and a common practice on a trading desk for a robust calibration of volatilities against market data. In practice, a very common parametrization is to take the volatility of each forward as constant over its own accrual interval $[T_i,T_{i+1}]$ or over coarser maturity buckets~\cite{BrigoM2006,AndersenP2010,Rebonato2002,AmetranoB2013}. The covariance structure can be written as 
\begin{align}
\mathrm{Cov}(dL_i(t),dL_j(t))
&= L_i(t)L_j(t)\,\mathrm{Cov}\Big(\bm{\sigma}_i(t)\cdot d\bm{W}_t,\ \bm{\sigma}_j(t)\cdot d\bm{W}_t\Big) \nonumber\\
&= L_i(t)L_j(t)\sum_{k=1}^d \sigma_{ik}(t)\sigma_{jk}(t)\,dt =: L_i(t)L_j(t)\,c_{ij}(t)\,dt,
\label{eq:Cov-LiLj}
\end{align}
where
\begin{equation}
c_{ij}(t)
:= \bm{\sigma}_i(t)\cdot\bm{\sigma}_j(t)
= \sum_{k=1}^d \sigma_{ik}(t)\sigma_{jk}(t).
\label{eq:C-entries}
\end{equation}
Thus, the $(N+1)\times (N+1)$ matrix $C(t) := (c_{ij}(t))_{0\le i,j\le N}$ is the instantaneous covariance matrix of the instantaneous {\it log-LIBOR} forward rated in factor representation. Under our assumption that the volatilities are piecewise constant, this and the instantaneous correlation matrix are also piecewise constant matrices.

For pricing of derivatives, it is convenient to move under a single measure. In practice, the \emph{terminal bond numeraire}
\[
N(t) := P(t,T_{N+1})
\]
is considered and the associated $T_{N+1}$-forward (terminal) measure $\mathbb{Q}^{N+1}$, defined by the property that any discounted asset price $X(t)/N(t)$ is considered to be a $\mathbb{Q}^{N+1}$-martingale. Under $\mathbb{Q}^{N+1}$, the forward LIBOR forward rates lose the nice propoerty of being driftless and satisfy the SDEs of the form
\begin{equation}
dL_j(t)
= L_j(t)\left(\mu_j(t,L(t))\,dt
+ \,\bm{\sigma}_j(t)\cdot d\bm{W}_t^{N+1}\right),
\qquad j=0,\dots,N-1,
\label{eq:LMM-terminal-measure-general}
\end{equation}
where $\bm{W}^{N+1}$ is a $d$-dimensional Brownian motion under $\mathbb{Q}^{N+1}$, and $\mu_j(t,L)$ are the drifts induced by the change of measure from the individual forward measures $\mathbb{Q}^{j+1}$ to $\mathbb{Q}^{N+1}$. The drift $\mu_j(t,L)$ is known explicitly and is, in general, a nonlinear function of the later forward rates~\cite{BraceGM1997,BrigoM2006,Fries2007,Golle2001,Lesniewski2007}: 
\begin{equation}
\mu_j(t,L(t))
= - \sum_{l=j+1}^{N-1} 
\frac{\alpha_l L_l(t)}{1+\alpha_l L_l(t)}\,
\bm{\sigma}_j(t)\cdot\bm{\sigma}_l(t),
\quad \text{for all} \quad j = \{0,\dots,N-1\}
\qquad \text{and} \quad \mu_{N}(t,L(t)) = 0.
\label{eq:mu-terminal-vector}
\end{equation}
It implies that under the terminal measure, the drift of $L_j$ depends only on later forwards $L_l$ with $l>j$, and the dependence is nonlinear via the weights
\[
w_l(L_l(t)) := \frac{\alpha_l L_l(t)}{1+\alpha_l L_l(t)}.
\]
Writing the instantaneous covariance matrix as in \eqref{eq:Cov-LiLj}, we summarize the dynamics under $\mathbb{Q}^{N+1}$ as
\begin{equation}
dL_j(t)
= L_j(t)\left(\mu_j(t,L(t))\,dt
+ \bm{\sigma}_j(t)\cdot d\bm{W}_t^{N+1}\right),
\qquad
\mathrm{Cov}\big(dL_i(t),dL_j(t)\big) = L_i(t)L_j(t)c_{ij}(t)dt.
\label{eq:LMM-terminal-summary}
\end{equation}

In the exact LMM case, under a common measure such as $\mathbb{Q}^{N+1}$, the drifts $\mu_j(t,L)$ are clearly non-linear functions of the LIBOR forward rates, and therefore, depend on $L$ as seen in \eqref{eq:mu-terminal-vector}. This state dependence complicates both analytic and numerical methods. A widely used approximation in the LMM literature is to replace this quantity by its intial value. Concretely, one sets 
\begin{equation}
\mu_j(t,L(t))
\approx \mu_j(t) := \mu_j\big(t,L(0)\big),
\label{eq:frozen-drift}
\end{equation}
so that the drift becomes a deterministic function of time only. This approximation dates back to Brace, Gatarek and Musiela~\cite{BraceGM1997} for swaption pricing, and is identified as the ``standard approximation'' in later work by Hunter, J\"ackel and Joshi~\cite{HunterJJ2001}, Papapantoleon~\cite{Papapantoleon2011}, and Papapantoleon and Skovmand~\cite{PapapantoleonS2011}, among others. Brigo and Mercurio~\cite{BrigoM2005} analyze the distributional distance between the LMM and the swap market model, and provide numerical evidence that this approximation yeilds swaption price that is accurate for many European swaptions. We call this model {\it approximate-LMM} and will derive a heat equation PDE for the valuation of European swaption under the model dynamics
\begin{equation}
dL_j(t)
= L_j(t)\left(\mu_j(t)\,dt
+ \bm{\sigma}_j(t)\cdot d\bm{W}_t^{N+1}\right),
\qquad
\mathrm{Cov}\big(dL_i(t),dL_j(t)\big) = L_i(t)L_j(t)c_{ij}(t)dt
\label{eq:approx-LMM-terminal-summary}
\end{equation}
where the volatility vector $\bm{\sigma}_j(t)$ satisfy \eqref{eq:piecewise-const-vol} for all $j$.

\subsubsection{European swaptions}
Before defining European swaptions, we first need to define swaps. A {\it vanilla swap} is a contract between two parties that exchange recurring payments based on a notional amount (henceforth, assume to be one unit of currency), where one party pays based on a fixed rate and the other part pays based on a floating rate. Assume that the payment dates are $T_0 < T_1 < \dots < T_N$,
and year-fractions $\alpha_k := T_k - T_{k-1}$ where $k = \{1,\dots,N\}$. The contract starts at $T_0$ and ends at $T_N$. Assuming $K$ to be the fixed rate. the party paying the fixed rate pays $\alpha_k K$ at the $k$-th payment date. This implies that the present value of the fixed leg is
\[
\text{PV}_{\mathrm{fixed}}(t) = K \sum_{k=1}^{N} \alpha_k P(t,T_k) = K A(t) \qquad \text{where} \qquad
A(t) = \sum_{k=1}^{N} \alpha_k P(t,T_k)
\]
is called the {\it swap annuity}. The floating rate pays $\alpha_k L(T_{k-1}, T_k)$ where $L(T_{k-1}, T_k)$ is the forward rate fixed at $T_{k-1}$ for borrowing between the time period $(T_{k-1}, T_k)$. Given that the forward rate must give the same payoff as investing in the zero-coupon bond maturing at $T_k$ (to remove arbitrage opportunities), we have the relationship
\[
1+\alpha_k L(T_{k-1},T_k) = \frac{1}{P(T_{k-1},T_k)}.
\]
We now compute the time-$t$ present value of the floating payment at $T_k$:
\[
\text{PV}_k(t) = P(t,T_k) \,
\mathbb{E}^{\mathbb{Q}}
\big[
\alpha_k L(T_{k-1},T_k)
\mid
\mathcal{F}_t
\big] = P(t,T_k) \left(\frac{P(t,T_{k-1})}{P(t,T_k)} - 1\right) =
P(t,T_{k-1}) - P(t,T_k).
\]
The latter equality follows from a no-arbitrage principle (refer to Brigo and Mercurio for detail). Finally, the total floating let payment is 
\[
\text{PV}_{\mathrm{float}} = \sum_{k=1}^N (P(t,T_{k-1}) - P(t,T_k)) = P(t,T_0) - P(t,T_N).
\]
If $t=T_0$, then $P(t,T_0) = 1$, and hence, $\text{PV}_{\mathrm{float}} = 1 - P(t, T_N)$. The value of the swap is 
\[
\text{Swap}(t) = \text{PV}_{\mathrm{float}} - \text{PV}_{\mathrm{fixed}} = 
1 - P(t, T_N) - KA(t).
\]

Now we move to first define European style swaptions. A {\it payer} swaption with expiry at $T_0$ gives the party to enter into a vanilla swap starting at $T_0$ where the party pays the fixed rate $K$ and received floating rate. The payoff of such a swaption at the maturity date $T_0$ (which coincides with the start date of the swap agreement) is 
\[
V_{\mathrm{payer}}(T_0) = \left(\text{Swap}(T_0)\right)^+ = 
\left(1 - P(T_0, T_N) - KA(T_0)\right)^+.
\]
A {\it receiver} swaption is just the opposite with the payoff function
\[
V_{\mathrm{receiver}}(T_0) = \left(KA(T_0) - (1 - P(T_0, T_N))\right)^+.
\]
The relationship between the payer and receiver swaption is that the sum of their payoffs is the value of the vanilla swap. So, assuming all the parameters remain same, finding the price of one of them is equivalent to finding the price of the other. Hence, the rest of this section will be devoted to deriving the PDE for the payer swaption. Defining the forward swap rate as $S(t) = (1-P(t,T_N)/A(t)$ (swap rate viewed at time $t$), we can rewrite the payer swaption payoff as
\[
V_{\mathrm{payer}}(t) = A(t)(S(t)-K)^+
\]
to resemble it more like an option payoff. Given that $A(t)$ is a function of price of ZCBs, one can rewritw the option payoff as
\begin{equation}
V_{\mathrm{payer}}(t) = g(L(t)),
\qquad \text{where} \qquad
g:\mathbb{R}_+^{N+1}\to\mathbb{R}_+,
\label{eq:swaption-payoff-G}
\end{equation}
is an explicit function of the LIBOR forward rates $L(t) := (L_0(t),\dots,L_N(t))$. Note that this function is a non-convex function in its arguments.

Let $V(t)$ denote the time-$t$ arbitrage-free price of the European payer swaption with payoff defined in \eqref{eq:swaption-payoff-G}. Using $N(t)=P(t,T_{N+1})$ as numeraire, we have
\begin{equation}
V(t) = P(t,T_{N+1})\,
\mathbb{E}^{N+1}\!\left[\left.
\frac{\Pi}{P(T_0,T_{N+1})}
\ \right|\ \mathcal{F}_t\right],
\qquad 0\le t\le T_0.
\label{eq:swaption-value-numeraire}
\end{equation}
Define the discounted value
\begin{equation}
v(t,L)
:= \frac{V(t)}{P(t,T_{N+1})}
= \mathbb{E}^{N+1}\!\left[\left.
\frac{V_{\mathrm{payer}}(T_0)}{P(T_0,T_{N+1})}
\ \right|\ L(t)=L\right]
= \mathbb{E}^{N+1}\!\left[ g(L(T_0)) \mid L(t)=L\right].
\label{eq:v-def}
\end{equation}
By construction, $\{v(t,L(t))\}_{t\in[0,T_0]}$ is a $\mathbb{Q}^{N+1}$-martingale. Under sufficient regularity, using It\^o's formula, the martingale property and Feynman-Kac formula,
\begin{equation}
\frac{\partial v}{\partial t}(t,L)
+ \sum_{j=0}^{N-1} L_j \mu_j(t) \frac{\partial v}{\partial L_j}(t,L)
+ \frac{1}{2}\sum_{i,j=0}^{N-1} L_i L_j c_{ij}(t)
\frac{\partial^2 v}{\partial L_i \partial L_j}(t,L) = 0,
\quad t<T_0,
\label{eq:swaption-pde-L}
\end{equation}
with terminal condition
\begin{equation}
v(T_0,L) = g(L),
\label{eq:swaption-terminal-L}
\end{equation}
where $g$ is defined in \eqref{eq:swaption-payoff-G}. The approach to transform the above PDE to a heat equation follows the following three steps:
\begin{mylist}{\parindent}
    \item [1.] A logarithmic transformation to remove the state dependence in the diffusion term and the first-order term similar to what we do for the equity options.
    \item [2.] A diagonalization of the coefficients of the resulting diffusion terms to remove the second-order mixed derivatives. We can do this by constant orthonormal matrices under the assumption that the volatilities are piecewise constant matrix.
    \item [3.] A gauge transformation to remove the first-order term and reduce the PDE obtained in the previous step to a heat equation. We can do this because we assumed that the coefficient of the first-order term is time-dependent and not state-dependent.
\end{mylist}

\noindent \paragraph{Removing state dependence} Now let's move to removing the state dependence of the second-order terms. The standard trick that one employs to transform Black-Scholes PDE to the heat equation will be our first step. More concretely, set $X_j := \log L_j$ for all $j\in \{0, \dots, N\}$ and define
\begin{equation}
u(t,x) := v\big(t,e^x\big)
\qquad \text{where} \qquad x\in\mathbb{R}^{N}.
\label{eq:u-def}
\end{equation}
Since $x_i$ only depends on $L_i$, we have
\[
\frac{\partial v}{\partial L_j}
= \frac{1}{L_j} \frac{\partial u}{\partial x_j}
\qquad \text{and} \qquad
\frac{\partial^2 v}{\partial L_i \partial L_j}
= -\frac{1}{L_i^2}\frac{\partial u}{\partial x_i}
+ \frac{1}{L_iL_j}\frac{\partial^2 u}{\partial x_i \partial x_j}.
\]
Substituting into the second-order term of \eqref{eq:swaption-pde-L},
\begin{align*}
\sum_{i,j=0}^{N-1} L_i L_j c_{ij}(t) \frac{\partial^2 v}{\partial L_i \partial L_j}
&= -\sum_{i=0}^{N-1} c_{ii}(t)\,\frac{\partial u}{\partial x_i}
+ \sum_{i,j=0}^{N-1} c_{ij}(t)\,\frac{\partial^2 u}{\partial x_i \partial x_j}.
\end{align*}
The drift term becomes
\[
\sum_{j=0}^{N-1} L_j \mu_j(t) \frac{\partial v}{\partial L_j}
= \sum_{j=0}^{N-1} L_j \mu_j(t) \frac{\partial u}{\partial x_j}
= \sum_{j=0}^{N-1} \mu_j(t)\,\frac{\partial u}{\partial x_j}.
\]
The time derivative remains unchanged, and hence, the PDE \eqref{eq:swaption-pde-L} transforms into
\begin{equation}
\frac{\partial u}{\partial t}(t,x)
+ \sum_{j=0}^{N-1} \left(\mu_j(t) - \frac{c_{jj}(t)}{2}\right)  \frac{\partial u}{\partial x_j}(t,x)
+ \frac{1}{2}\sum_{i,j=0}^{N-1} c_{ij}(t) \frac{\partial^2 u}{\partial x_i \partial x_j}(t,x) = 0.
\label{eq:swaption-pde-X-exact}
\end{equation}

\noindent \paragraph{Removing mixed derivative terms} Our next step is the removal of the mixed second-order derivative terms. Under the assumption that the volatilies are piecewise constant, we use a single orthonormal transformation within each time slab where the volatilies are constant to diagonalize the diffusion term. More concretely, in a subinterval $(\tau_k,\tau_{k+1}]$ on which $C(t)\equiv C^{(k)}$ is constant, we have
\begin{equation}
C^{(k)} = Q^{(k)} \Lambda^{(k)} \big(Q^{(k)}\big)^\top,
\label{eq:C-spectral}
\end{equation}
where $Q^{(k)}$ is an orthonormal matrix and $\Lambda^{(k)}=\mathrm{diag}\left(\lambda_1^{(k)},\dots,\lambda_{N}^{(k)}\right)$ is a diagonal matrix with strictly positive entries $\lambda_m^{(k)}$. Let 
\begin{equation}
y := \big(Q^{(k)}\big)^\top x 
\qquad \text{such that} \qquad = Q^{(k)} y,
\label{eq:y-def}
\end{equation}
and define $w(t,y) := u\big(t,Q^{(k)}y\big)$.
Using the chain rule, we have
\[
\nabla_x u = Q^{(k)} \nabla_y w
\qquad \text{and} \qquad 
\sum_{i,j=0}^{N-1} c_{ij}^{(k)} \frac{\partial^2 u}{\partial x_i \partial x_j}
= \sum_{m=1}^{N} \lambda_m^{(k)} \frac{\partial^2 w}{\partial y_m^2}
\]
while the drift term transforms as
\[
\sum_{j=0}^{N-1} \hat{\mu}_j(t) \frac{\partial u}{\partial x_j}
= \hat{\mu}(t)^\top \nabla_x u
= \big(Q^{(k)}{}^\top \hat{\mu}(t)\big)^\top \nabla_y w
= \sum_{m=1}^{N} \hat{\mu}_m^{(k)}(t) \frac{\partial w}{\partial y_m},
\]
where
\[
\hat{\mu}^{(k)}(t) := \big(Q^{(k)}\big)^\top \hat{\mu}(t).
\]
Thus, on $(\tau_k,\tau_{k+1}]$, the resulting PDE becomes
\begin{equation}
\frac{\partial w}{\partial t}(t,y)
+ \sum_{m=1}^{N} \hat{\mu}_m^{(k)}(t)\, \frac{\partial w}{\partial y_m}(t,y)
+ \frac{1}{2}\sum_{m=1}^{N} \lambda_m^{(k)} \frac{\partial^2 w}{\partial y_m^2}(t,y) = 0.
\label{eq:swaption-pde-Y}
\end{equation}

\noindent \paragraph{Removing first-order terms} To finally remove the first-order terms and obtain a pure heat-type equation, we perform a gauge (exponential) transform. On $(\tau_k,\tau_{k+1}]$, let
\begin{equation}
w(t,y) = \exp\!\left(a^{(k)}(t) + \ip{b^{(k)}(t)}{ y}\right)\, z(t,y),
\label{eq:gauge-ansatz}
\end{equation}
where $a^{(k)}(t) \in \real$ and $b^{(k)}(t) \in \real^N$ are to be chosen later on. Substituting \eqref{eq:gauge-ansatz} in \eqref{eq:swaption-pde-Y}, would introduce zero-order and first-order terms in $z(t,y)$, while the diffusion term will remain unchanged. Since we wish to obtain a pure heat-type equation, we put their respective coefficients to $0$ and solve for $a^{(k)}(t)$ and $b^{(k)}(t)$. Once we have that, our swaption PDE reduces to
\begin{equation}
\frac{\partial z}{\partial t}(t,y)
= \frac{1}{2}\sum_{m=1}^{N} \lambda_m^{(k)} \frac{\partial^2 z}{\partial y_m^2}(t,y),
\qquad t\in(\tau_k,\tau_{k+1}],\ t<T_0.
\label{eq:heat-diagonal}
\end{equation}
This is a multi-dimensional \emph{diagonal heat equation} with constant diffusion coefficients $\lambda_m^{(k)}$ on each time slab. The terminal condition for $z$ at $t=T_0$ is obtained from the swaption payoff $g(\exp(x)$ by undoing the transformations 
\[
L\mapsto X=\log L, \qquad
X\mapsto y=(Q^{(k)})^\top X, \qquad \text{and} \qquad
w\mapsto z
\]
on \eqref{eq:gauge-ansatz}. Globally in time, one proceeds slab by slab from $T_0$ down to $0$, matching the transformed solution across the partition times $\{\tau_k\}$.

\subsubsection{Bermudan swaptions}
Let Bermudan exercise times be $0<t_1<\cdots<t_R=T_0$. Form a unified increasing grid
\[
0=s_0<s_1<\cdots<s_K=T_0,
\qquad
\{s_k\}=\{t_r\}\cup\{\tau_m\}
\]
(duplicates merged). We compute backward from $s_k$ to $s_0$. Fix $k\in\{0,\dots,K-1\}$. On $(s_k,s_{k+1})$ the coefficients are constant (the interval lies inside a unique volatility slab), so \eqref{eq:heat-diagonal} applies with some $\Lambda^{(k)}$ and the solution is an explicit Gaussian convolution:
\begin{equation}
z^{(k)}(s_k,y)=
:=\int_{\mathbb{R}^d}G_k(y-\xi;\Delta_k)\,z^{(k)}(s_{k+1},\xi)\,d\xi,
\qquad \Delta_k:=s_{k+1}-s_k,
\label{eq:heat_kernel_step}
\end{equation}
where
\begin{equation}
G_k(u;\Delta)=\prod_{m=1}^d\frac{1}{\sqrt{2\pi \lambda_m^{(k)}\Delta}}
\exp\!\left(-\frac{u_m^2}{2\lambda_m^{(k)}\Delta}\right).
\label{eq:heat_kernel}
\end{equation}
Thus, between grid points the propagation is explicit (but generally high-dimensional). n an interval where the transform parameters are fixed (a given slab), we have the invertible chain
\[
L \mapsto x=\log L \mapsto y=Q_k^\top x,
\qquad
u(t,x)=\exp\!\big(a_k(t)+b_k(t)^\top y\big)\,z^{(k)}(t,y),
\qquad
v(t,L)=u(t,\log L).
\]
Equivalently, at a fixed time $t$ in slab $k$,
\begin{equation}
z^{(k)}(t,y)=\exp\!\big(-a_k(t)-b_k(t)^\top y\big)\,v\big(t,L(y)\big),
\qquad
L(y)=\exp(Q_k y),
\label{eq:z_from_v}
\end{equation}
and
\begin{equation}
v\big(t,L(y)\big)=\exp\!\big(a_k(t)+b_k(t)^\top y\big)\,z^{(k)}(t,y).
\label{eq:v_from_z}
\end{equation}
Now, there are two points to consider: (a) the Bermudan exercise dates, and (b) the point where the volatilities and correlation structure changes. We proceed below to work out both scenarios.

\noindent \paragraph{Exercise dates:}~At an exercise time $s_k=t_r$, the Bermudan condition in physical (undiscounted) value is
\[
V(s_k^-,L)=\max\big(V_{\mathrm{intr}}(s_k,L),\,V_{\mathrm{cont}}(s_k^+,L)\big).
\]
In discounted form $v=V/N$ (under any chosen numeraire $N$), this becomes
\begin{equation}
v(s_k^-,L)=\max\big(g_{\mathrm{intr}}(s_k,L),\,v(s_k^+,L)\big),
\qquad
g_{\mathrm{intr}}(s_k,L):=\frac{V_{\mathrm{intr}}(s_k,L)}{N(s_k)}.
\label{eq:exercise_v}
\end{equation}
Using \eqref{eq:z_from_v}, the corresponding jump condition for $z^{(k)}$ at $s_k$ is
\begin{equation}
z^{(k)}(s_k^-,y)
=
\max\!\left(
\exp\!\big(-a_k(s_k)-b_k(s_k)^\top y\big)\,g_{\mathrm{intr}}(s_k,L(y)),
\ \ z^{(k)}(s_k^+,y)
\right).
\label{eq:exercise_z}
\end{equation}
This is the early-exercise ``free boundary'' condition expressed in the heat coordinates; the nonlinearity (the max) is the essential Bermudan feature and is treated by backward induction (often combined with regression/approximation in high dimensions) \cite{Andersen2000,Glasserman2003}.

\noindent \paragraph{Volatility/ correlation change dates:}~Moving on to the inflection point where the volatilities change ($s_k=\tau_m$), the physical value $v(s_k,L)$ is continuous in time (no jump purely from coefficient relabeling):
\begin{equation}
v(s_k^-,L)=v(s_k^+,L), \qquad \forall L.
\label{eq:continuity_v}
\end{equation}
However, the transforms differ on the left and right because $C^{(k)}$ changes and so does $(Q_k,\Lambda_k,a_k,b_k)$. Let the representation on $(s_{k-1},s_k)$ use parameters $(Q_{k-1},a_{k-1},b_{k-1})$ and on $(s_k,s_{k+1})$ use $(Q_k,a_k,b_k)$. For a fixed physical state $L$, define
\[
x=\log L,\quad y^-:=Q_{k-1}^\top x,\quad y^+:=Q_k^\top x.
\]
Then \eqref{eq:continuity_v} and \eqref{eq:v_from_z} imply
\[
\exp\!\big(a_{k-1}(s_k)+b_{k-1}(s_k)^\top y^-\big)z^{(k-1)}(s_k,y^-)
=
\exp\!\big(a_k(s_k)+b_k(s_k)^\top y^+\big)z^{(k)}(s_k,y^+).
\]
Let the orthogonal matrix $R_k:=Q_k^\top Q_{k-1}$ so that $y^+=R_k y^-$. Then the interface map for the heat-representation is
\begin{equation}
z^{(k-1)}(s_k,y^-)
=
\exp\!\Big(a_k(s_k)-a_{k-1}(s_k)+b_k(s_k)^\top (R_k y^-)-b_{k-1}(s_k)^\top y^-\Big)\,
z^{(k)}(s_k,R_k y^-).
\label{eq:interface_map}
\end{equation}
This interface condition is the precise way to ``propagate across slabs'' when the diagonalizing rotation changes (which is generic unless covariance matrices commute across time) \cite{Rebonato2002,BrigoM2006,Glasserman2003}.

\subsubsection{Cancelable swaps}
A cancelable swap is a swap where one party (usually the fixed-rate payer) has the right to terminate the swap early at a discrete set of dates. If the discrete set if dates when the party can exercise this optionality of terminating the swap coincides exactly with the exercise dates of a Bermudan swaption, then we obtain the following relationship: the value of a cancelable swap is exactly the difference between the value of the swap and the value of the Bermudan swaption. This relationship implies that pricing the cancelable swap can be done by pricing the Bermudan swaption.

\subsection{Hybrid derivatives}
We now move to derivatives that depend on the dynamics of interest-rates and foreign-exchange (FX) rates. Foreign-exchange market is a highly liquid market where the underlying asset is typically the exchange rate between two fiat currencies. A Black-Scholes type SDE  models the exchange rate between two currencies, where each currency's value depends on their respective interest-rate dynamics. It is common to call one currency as domestic currency and the other as foreign currency. Thus, their rate dynamics are usually classified in domestic leg and foreign leg, where it is understood that the FX spot price is given by the Black Scholes equation as we describe below.

\subsubsection{Cross-currency swaptions}
Similar to vanilla swaptions discussed before, a cross-currency swaption gives a party right to enter into a cross-currency swap where the trading parties enter into a swap agreement to pay each other based on floating rate governed by the interest-rate in their respective currencies. This section develops the pricing PDE for a cross-currency swaption (CCS swaption) under a domestic money-market numeraire. The setup includes domestic and foreign LMMs, FX dynamics, quanto drift adjustments, and market-practical assumptions of drift approximation and piecewise-constant volatility. These assumptions allow the CCS swaption pricing PDE to be transformed into a heat-type PDE on each volatility slab, enabling efficient numerical propagation. The framework builds on the single-currency LMM of Brace, Gatarek and Musiela and its multi-currency extensions, as discussed in various references~\cite{BraceGM1997,MiltersenSS1997,Rebonato2002,BrigoM2006,AndersenP2010,BennerZ2008,SchoenmakersH2011}.

To begin with, let us first describe the domestic and foreign LMMs. Consider domestic tenor dates 
\[
0=T^d_0<T^d_1<\cdots<T^d_{N_d+1},\qquad \alpha^d_i=T^d_{i+1}-T^d_i,
\]
and foreign tenor dates
\[
0=T^f_0<T^f_1<\cdots<T^f_{N_f+1},\qquad \alpha^f_j=T^f_{j+1}-T^f_j.
\]
Domestic zero-coupon bonds are denoted $P^d(t,T)$, foreign zero-coupon bonds $P^f(t,T)$. The forward domestic LIBORs are defined by the standard single-currency relation
\[
L_i^d(t)=\frac{1}{\alpha^d_i}\left(\frac{P^d(t,T^d_i)}{P^d(t,T^d_{i+1})}-1\right),
\qquad
P^d(t,T^d_{i+1})=\frac{P^d(t,T^d_i)}{1+\alpha^d_i L^d_i(t)},
\]
and similarly for the foreign $L_j^f(t)$. These definitions allow all domestic and foreign bond prices to be constructed recursively from forward rates. The single-currency LMM structure follows the classical BGM framework \cite{BraceGM1997,MiltersenSS1997,Rebonato2002,BrigoM2006,HullW2000a}.

Under each domestic forward measure $\mathbb{Q}^{d,i+1}$ with numeraire $P^d(t,T^d_{i+1})$, the domestic forward rate is driftless:
\[
dL_i^d(t)=L_i^d(t)\,\sigma_i^d(t)\,dW^{d,i+1}_i(t),
\qquad 
\ip{dW^{d,i+1}_i}{dW^{d,i+1}_k}=\rho^{dd}_{ik}(t)\,dt,
\]
and similarly for the foreign forwards under their own forward measures $\mathbb{Q}^{f,j+1}$. Volatilities and correlations may differ between domestic and foreign curves, and empirical studies support multi-factor correlation structures across forward rates and currencies \cite{Rebonato2002,AndersenP2010,BennerZ2008}.

Now let us move to explain the FX dynamics between the currency pair. Let $S(t)$ denote the FX rate in domestic currency per foreign currency unit. Under the domestic money-market numeraire
\[
B^d(t)=\exp\!\left(\int_0^t r^d(u)\,du\right),
\]
the domestic risk-neutral measure $\mathbb{Q}^d$ is defined so that any domestic-denominated asset discounted by $B^d$ is a martingale. The foreign money-market account $B^f(t)$ is defined analogously using the foreign short rate $r^f(t)$. A standard specification for FX under its foreign risk-neutral measure is
\[
\frac{dS(t)}{S(t)} = (r^d(t)-r^f(t))\,dt + \sigma^S(t)\cdot dW^f(t).
\]
After changing to the domestic measure $\mathbb{Q}^d$, the FX drift picks up quanto terms that depend on the correlation between FX and the foreign numeraire, see \cite{BrigoM2006,AndersenP2010,Jamshidian1997,SchoenmakersC2003,BrigoBN2015}. In particular, one can write under $\mathbb{Q}^d$:
\[
\frac{dS(t)}{S(t)}=\mu^S(t)\,dt+\sigma^S(t)\cdot dW^d(t),
\]
where
\[
\mu^S(t)
=
r^d(t)-r^f(t)-\sum_{j=1}^{N_f}\frac{\alpha^f_j L^f_j(t)}{1+\alpha^f_j L^f_j(t)}\,\sigma^f_j(t)\cdot\sigma^S(t),
\]
and $W^d$ is an $m$-dimensional Brownian motion under $\mathbb{Q}^d$.

\noindent\textbf{Explicit Domestic-Measure Drifts for Domestic and Foreign LMMs.}
The domestic forward-rate drifts under the domestic money-market measure follow the classical BGM structure. Under $\mathbb{Q}^d$, one can write
\[
dL_i^d(t)=L_i^d(t)\left(\mu_i^d(t)\,dt+\sigma_i^d(t)\cdot dW^d(t)\right),
\]
with drift
\[
\mu_i^d(t)
=
\sum_{k=i+1}^{N_d}
\frac{\alpha^d_k L_k^d(t)}{1+\alpha^d_k L_k^d(t)}\,\sigma_i^d(t)\cdot\sigma_k^d(t).
\]
This expression, derived via Radon--Nikodym arguments for measure changes between forward measures and the money-market measure, appears in \cite{BraceGM1997,Jamshidian1997,Rebonato2002,BrigoM2006,AndersenP2010} and is central to LMM implementations.

For foreign rates, under their own forward measures $dL_j^f(t)=L_j^f(t)\,\sigma_j^f(t)\cdot dW^{f,j+1}(t)$, switching to the domestic money-market measure produces both the usual LMM drift and a quanto term arising from FX. Under $\mathbb{Q}^d$,
\[
dL_j^f(t)=L_j^f(t)\left(\mu_j^f(t)\,dt+\sigma_j^f(t)\cdot dW^d(t)\right),
\]
where
\[
\mu_j^f(t)
=
\sum_{\ell=j+1}^{N_f}
\frac{\alpha^f_\ell L_\ell^f(t)}{1+\alpha^f_\ell L_\ell^f(t)}\,\sigma_j^f(t)\cdot\sigma_\ell^f(t)
\;+\;
\sigma_j^f(t)\cdot\sigma^S(t).
\]
The second term is the quanto drift adjustment coupling foreign forward rates to FX volatility. Systematic derivations of these drifts are given in \cite{BrigoM2006,AndersenP2010,Jamshidian1997,SchoenmakersC2003,SchoenmakersH2011,BennerZ2008}.

\bigskip

\noindent\textbf{Piecewise-Constant Volatility and Frozen Drift.}
Market calibration of IR and FX term structures is performed in discrete tenor buckets; volatility surfaces are interpolated piecewise in time or variance. Accordingly, banks implement piecewise-constant volatilities for all stochastic factors, as documented in \cite{Rebonato2002,BrigoM2006,AndersenP2010,HullW2000a,HullW2000b,CarrW2007,BraceGM1997}. FX implied volatility quotes (ATM, risk reversal, butterfly) are specified at standard maturities (e.g.\ 1M, 3M, 6M, 1Y, 2Y, 5Y), and piecewise-constant instantaneous volatility within these buckets preserves integrated variance relevant for CCS products. Because CCS swaptions are primarily driven by medium- and long-dated IR exposures rather than short-dated FX micro-structure, the exact short-time shape of the FX instantaneous volatility curve is less important than the correct time-bucketed variance and correlations.

The exact domestic-measure drifts $\mu_i^d(t)$ and $\mu_j^f(t)$ are nonlinear functions of all future forwards and FX. The frozen-drift approximation replaces these state-dependent drifts with deterministic time functions, typically by freezing the LIBORs and FX in the drift at their initial (or forward) levels. This method, justified numerically in \cite{HunterJJ2001,HullW2000a,HullW2000b} and extended in cross-currency LMMs \cite{BennerZ2008,SchoenmakersH2011}, improves tractability while preserving calibration accuracy for swaption-like payoffs. Under frozen drift and piecewise-constant volatilities, the drift coefficients become functions of time only, and the diffusion matrix of log-variables becomes piecewise constant in time, exactly the regime in which heat-equation reductions are most effective.

\bigskip

\noindent\textbf{CCS Swaption Payoff in Domestic Currency.}
At option maturity $T_0$, a payer CCS swaption gives the right to enter a cross-currency swap receiving the foreign leg and paying the domestic leg. Let the domestic leg value be expressed as a deterministic function of domestic forwards,
\[
V^d_{\text{leg}}(T_0)=F^d(L^d(T_0)),
\]
and the foreign leg value in foreign currency as
\[
V^f_{\text{leg}}(T_0)=F^f(L^f(T_0)).
\]
Its domestic currency value is $S(T_0)\,F^f(L^f(T_0))$. The CCS swaption payoff is therefore
\[
\Pi=\Big(S(T_0)\,F^f(L^f(T_0))-F^d(L^d(T_0))\Big)^+,
\]
a function of the full state $(L^d,L^f,S)$. Denoting $X=(X^1,\dots,X^d)$ as the state vector, define $G(X)$ as the payoff function.

\bigskip

\noindent\textbf{Pricing Under the Domestic Numeraire.}
Under $\mathbb{Q}^d$ and numeraire $B^d$, the discounted price
\[
v(t,x)=\mathbb{E}^{\mathbb{Q}^d}[G(X(T_0))\mid X(t)=x]
\]
is a martingale. The domestic money-market numeraire is the natural choice because the payoff is denominated in domestic currency and foreign cashflows are converted using $S(t)$. Using this numeraire automatically incorporates FX--IR covariances and quanto effects through the drift adjustments described above \cite{BrigoM2006,AndersenP2010,Jamshidian1997,SchoenmakersC2003,BrigoBN2015}.

With frozen drift and piecewise-constant covariance over slabs $(\tau_k,\tau_{k+1}]$, the state evolves as
\[
\frac{dX^a_t}{X^a_t}=\bar\mu_a^{(k)}(t)\,dt+\sigma_a^{(k)}\cdot dW^d_t,
\quad t\in(\tau_k,\tau_{k+1}],
\]
where $\bar\mu_a^{(k)}(t)$ is deterministic and $\sigma_a^{(k)}$ are constant vectors on the slab. The backward PDE for $v$ is
\[
\partial_t v+\sum_a x^a\bar\mu_a(t)\partial_{x^a}v
+\frac12\sum_{a,b}x^a x^b c_{ab}(t)\partial^2_{x^a x^b}v=0,
\qquad v(T_0,x)=G(x),
\]
where $c_{ab}(t)=\sigma_a(t)\cdot\sigma_b(t)$ and $C(t)=[c_{ab}(t)]$ is piecewise constant in time.

\bigskip

\noindent\textbf{Log Transform.}
Define log-coordinates $y^a=\log x^a$ and $u(t,y)=v(t,e^y)$. The Lamperti transform removes the multiplicative diffusion factors. Using standard chain-rule calculations, one obtains
\[
\partial_t u+\sum_a \tilde\mu_a(t)\partial_{y^a}u
+\frac12\sum_{a,b}c_{ab}(t)\partial^2_{y^a y^b}u=0,
\]
with
\[
\tilde\mu_a(t)=\bar\mu_a(t)-\frac12 c_{aa}(t).
\]
On each slab $(\tau_k,\tau_{k+1})$, $C^{(k)}=[c_{ab}^{(k)}]$ is constant.

\bigskip

\noindent\textbf{Diagonalization.}
Let $C^{(k)}=Q^{(k)}\Lambda^{(k)}(Q^{(k)})^\top$ be the eigen-decomposition of the covariance matrix on slab $k$, with orthogonal $Q^{(k)}$ and diagonal $\Lambda^{(k)}=\mathrm{diag}(\lambda_1^{(k)},\dots,\lambda_d^{(k)})$. Setting
\[
z=(Q^{(k)})^\top y,\qquad w(t,z)=u(t,y)=u\big(t,Q^{(k)}z\big),
\]
removes mixed derivatives:
\[
\partial_t w+\sum_{m=1}^d \tilde\mu_m^{(k)}(t)\,\partial_{z_m} w
+\frac12\sum_{m=1}^d \lambda_m^{(k)}\,\partial^2_{z_m z_m} w=0,
\]
with $\tilde\mu^{(k)}(t)=(Q^{(k)})^\top\tilde\mu(t)$.

\bigskip

\noindent\textbf{Gauge Transform and Heat Equation.}
With the exponential ansatz
\[
w(t,z)=\exp(a_k(t)+b_k(t)\cdot z)\,\phi^{(k)}(t,z),
\]
and appropriate choice of $b_k(t)$ and $a_k(t)$ so that first-order and zero-order terms vanish, one arrives at a pure heat-type PDE on slab $k$:
\[
\partial_t \phi^{(k)}(t,z)=\frac12\sum_{m=1}^d \lambda_m^{(k)}\,\partial^2_{z_m z_m}\phi^{(k)}(t,z),
\qquad t\in(\tau_k,\tau_{k+1}),
\]
with terminal condition inherited from the payoff after all transformations. Thus, on each slab, the CCS swaption price is propagated by a multidimensional Gaussian convolution with eigenvalues $\{\lambda_m^{(k)}\}$.

\bigskip

\noindent\textbf{MTM versus Non-MTM CCS Swaptions.}
In a mark-to-market (MTM) CCS, the notional of the foreign leg is periodically reset using the prevailing FX rate. Non-MTM swaps maintain a fixed notional. This distinction changes the structure of $F^f(L^f)$: in MTM CCS, the foreign leg is less sensitive to long-dated FX levels because FX resets reduce accumulated FX mismatch, whereas non-MTM CCS exhibit stronger cross-currency optionality. Consequently, MTM CCS swaptions exhibit weaker FX--IR convexity effects and smaller quanto drift adjustments. Non-MTM CCS swaptions are more sensitive to modelling assumptions on FX volatility and correlations. These differences directly influence the drift terms $\bar\mu_a(t)$ and the terminal payoff $G(x)$, hence altering the PDE and its transformed terminal condition. Detailed discussions of MTM versus non-MTM structures can be found in \cite{AndersenP2010,Rebonato2002}.

\bigskip

\noindent\textbf{Why domestic numeraire} Readers may note that while in European swaptions, we worked with bond measure, in cross-currency swaption, we work with domestic money market numeraire. 
The  cross-currency swaption gives the right to enter a swap where domestic cash flows are exchanged for foreign cash flows, typically with FX conversion at each payment or at maturity. The value of this swaption is ultimately expressed in the domestic currency, since that’s what we're pricing in. Moreover, both domestic and foreign legs may have floating structures. The foreign leg must be converted to domestic via the FX rate, which is a stochastic process. The domestic bond measure is not a natural numeraire for the foreign leg, which has payouts at different times and in a different currency. Therefore, moving to domestic bond numeraire will cause complexities. The FX conversion introduces drift adjustments (quanto effects) that are cleaner when working under a risk-neutral measure for the domestic currency. Hence, we stay under the domestic money market account numeraire because it allows for unified valuation of all cash flows—including the FX-adjusted foreign leg—in one currency.

\bibliographystyle{alpha}
\bibliography{bibliography}

\end{document}